\documentclass[twocolumn,longbibliography,nofootinbib,aps,pra,superscriptaddress]{revtex4-2}
\usepackage{CJK}
\pdfoutput=1 
\usepackage{soul}
\usepackage{mathtools,amssymb,amsfonts,amsthm,physics,subcaption,graphicx,multirow,float,bm,xcolor}
\usepackage{caption}
\usepackage{ragged2e}
\usepackage[colorlinks,citecolor=blue,linkcolor=blue,urlcolor=blue, breaklinks=true]{hyperref}
\usepackage[capitalize]{cleveref}
\usepackage[shortlabels]{enumitem}
\setlist{
  listparindent=\parindent,
  parsep=0pt,
}
\usepackage{accents}
\counterwithout{figure}{section}
\usepackage[none]{hyphenat}
\usepackage{csquotes}
\usepackage{tikz}
\usetikzlibrary{calc, arrows.meta}
\usepackage{xcolor}
\usepackage{tikz-3dplot}
\usetikzlibrary{math}
\usetikzlibrary{calc}
\usetikzlibrary{fit,positioning}%
\usepackage{subcaption}
\newtheorem{rmk}{Remark}

\newtheorem{example}{Example}
\newtheorem*{lemma*}{Lemma}

\usepackage{color}

\usepackage{bm}
\usepackage{dsfont}
\setlist{nosep}
\newcommand{\ubar}[1]{\underaccent{\bar}{#1}}

\usepackage{thmtools}
\usepackage{thm-restate}

\newcommand{\nocontentsline}[3]{}
\let\oldaddcontentsline\addcontentsline
\newcommand{\tocless}[2]{%
  \let\addcontentsline=\nocontentsline#1{#2}
  \let\addcontentsline\oldaddcontentsline}

\newtheorem{theorem}{Theorem}
\newtheorem{lemma}{Lemma}
\newtheorem{definition}{Definition}

\newcommand{\ifusp}{Department of Mathematical Physics, Institute of Physics, University of São Paulo,\\ Rua do Matão 1371, São Paulo 05508-090, São Paulo, Brazil}

\begin{document}

\title{Contextual advantage implies limited distinguishability in any physical theory}

\author{Roberto D. Baldijão}
\email{rbaldijao@perimeterinstitute.ca}
\affiliation{Perimeter Institute for Theoretical Physics, 31 Caroline Street North, Waterloo, Ontario N2L 2Y5, Canada}
\affiliation{International Centre for Theory of Quantum Technologies, University of Gda\'nsk, 80-308 Gda\'nsk, Poland}
\author{Felipe A. Barretto}
\affiliation{International Centre for Theory of Quantum Technologies, University of Gda\'nsk, 80-308 Gda\'nsk, Poland}
\affiliation{\ifusp}
\author{Jaros\l{}aw K. Korbicz}
\affiliation{Center for Theoretical Physics, Polish Academy of Sciences, Aleja Lotnik\'{o}w 32/46, 02-668 Warszawa, Poland}

\begin{abstract}
A central question in quantum information is whether a given set of states can provide an advantage in some information-processing task. We establish a universal connection between two such questions: how well a set of states can be discriminated, and whether it can power tasks whose advantage stems from generalized contextuality. Working in the framework of generalized probabilistic theories (GPTs), which includes quantum and classical systems as special cases, we show that any set of states able to provide a nonclassical advantage in a contextuality-powered task must obey nontrivial upper bounds on the success probability of every state discrimination task using the full set. Contextual advantage therefore implies limited distinguishability, exposing a trade-off between two basic operational resources. Importantly, this is more than a statement about the idealized limit: perfect distinguishability is not needed to preclude contextuality. Our thresholds lie strictly below unity, so any set whose discrimination performance exceeds them while still imperfectly distinguishable is already guaranteed to admit a noncontextual explanation. The bounds follow from a simple geometric property, linear dependence of the state set, take a closed analytical form, and depend only on the prior of the task and the convex geometry of the states. As an illustration, we apply them to generalized parity-oblivious multiplexing, where sufficiently high success in the task implies that at least one sub-ensemble of the codebook cannot power any contextual advantage.
\end{abstract}
\maketitle

\raggedbottom

\section{Introduction}

A central task in quantum information is minimum-error state discrimination (SD), in which the goal is to determine in which state a system was prepared when it is drawn from a known set $\{s_k\}$ with prior probabilities $\{p_k\}$. SD tasks are central both to the foundations of quantum theory and to its technological applications and, arguably, to any information theory as well (see Refs.~\cite{Bae_StructureSDQuantum,Bae_StructureSDGPT,Bergou_SDReview2007,barnett_SDReview} and references therein).

Our first main result shows that a simple geometric property of a set of states-- namely, whether it is linearly dependent-- imposes nontrivial bounds on the achievable success probability in any SD task involving the full set. This goes beyond the familiar limiting case that perfect discrimination requires linear independence: our bounds imply that whenever the full set of states is linearly dependent, the optimal success probability in any SD task using all of them is \emph{bounded away from \(1\)}. The resulting bounds are straightforward to compute, relying only on the convex geometry of the state set, and can be used to assess whether a target discrimination performance is feasible for a given set of states. {Useful as these bounds might be on their own, their most significant consequence emerges only later, once we connect them to generalized contextuality (Sec.~\ref{sec:Contextuality}): they turn into a universal limitation on which state sets can power contextual advantage.}

Furthermore, our results are broadly applicable: they hold not only for quantum and classical systems but also for any physical theory described within the framework of generalized probabilistic theories (GPTs)~\cite{Mueller_GPTLesHouche2021,barrett_GPTsInformationProcessing,Plavala_GPTIntro2023,Chiribella_QuantumFromPrinciples2016,hardy2011ReconstructQT,hardy2001quantum}. This framework encompasses both classical and quantum theory as special cases, while also allowing for a wider class of hypothetical models. Even if such models do not describe nature, they help clarify what distinguishes quantum theory~\cite{hardy2001quantum,hardy2011ReconstructQT,MuellerMasanes_2011ReconstructingQT} and illuminate how operational principles constrain physical theories~\cite{Chiribella_Purification2010,Mueller_BitsymmetryDual2012,barnum2019stronglysymmetricspectralconvex,BarnumMueller_HigherOrderInterference_2014,Pfister2013_InfoPrinciplePolytopicTheories,Baldijao_2022}. Besides applying to this broader class of theories, our results are more general in an additional sense: the bounds associated with a state set $\{s_k\}$ remain valid for any other set $\{t_k\}$, possibly belonging to a different GPT, provided that $\{t_k\}$ is related to $\{s_k\}$ by an injective linear map, $t_k = M s_k$ for all $k$. This implies that the bounds are invariant under linear maps that preserve the relevant geometric relations among the states and thus apply to equivalence classes of GPT state sets.

This generality also allows us to connect our SD bounds to a central notion of nonclassicality in any GPT system, namely generalized contextuality~\cite{Spekkens_Contextuality2005,Spekkens_NegativityAndC2008}. Contextuality captures when observed statistics fail to admit a classical explanation in which operationally indistinguishable procedures are represented identically, and has a precise formulation in prepare-and-measure scenarios within the GPT framework~\cite{Schmid_2021_NCinGPTsystems}. It subsumes several notions of nonclassicality, including Bell nonlocality~\cite{Bell1964} and Kochen--Specker contextuality~\cite{KochenSpecker1967}, and has been identified as a key resource underlying advantages in a wide range of information-processing tasks, from communication~\cite{Spekkens_POM,Chailloux_2016} and cryptography to computation\cite{Howard2014,Raussendorf2013} and thermodynamics~\cite{Lostaglio_CandTheromLinearResponse2020}.
Our second main result shows that any set of states that gives rise to contextuality in prepare-and-measure experiments must obey our bounds on state discrimination. Thus, contextual advantage implies a fundamental limitation on distinguishability. Conversely, if a set of states ${s_k}$ is sufficiently distinguishable to violate our bounds, then any prepare-and-measure experiment using this set admits a classical explanation. This shows that nonclassical advantage is fundamentally  and universally incompatible with sufficiently high discriminability, in any GPT.

Before presenting these results, we briefly review the GPT framework for prepare-and-measure scenarios and the basics of state discrimination. In Sec.~\ref{sec:Technical} we show our technical results leading to the SD bounds and their properties, and in sec.~\ref{sec:Contextuality} we state our main conceptual result.

\section{Preliminaries}
\label{sec:Preliminaries}
{The scenarios we consider throughout this work are \emph{prepare-and-measure} scenarios. In such scenarios, a physical system is prepared through some procedure and is subsequently measured, producing an outcome (see Fig.~\ref{fig:PrepareAndMeasure}).} {Denoting preparation procedures by $P$ and measurement procedures by $M$, with outcomes labelled by $k$, repeating the experiment many times yields the statistics of outcomes conditioned on the procedures, $p(k\,|\,M,P)$. Different procedures may be indistinguishable at the level of these statistics: two preparations are \emph{operationally equivalent}, denoted $P\simeq P'$, if $p(k|M,P)=p(k|M,P')$ for every measurement event $(k|M)$; likewise, two measurement events are operationally equivalent, $(k|M)\simeq(k'|M')$, if they occur with the same probability for every preparation. In most of what follows, as is customary, we work directly with  the equivalence classes, represented by \emph{states} $s\leftrightarrow[P]$ and \emph{effects} $e\leftrightarrow[k|M]$; for instance, in quantum theory, an equivalence class of preparations of a qubit might be represented by the state $\ketbra{+}{+}$. The distinction between procedures within an equivalence classes will be needed again only when we discuss (non)contextuality, in Sec.~\ref{sec:Contextuality}.} {This minimal structure is at the heart of much of quantum information: state discrimination, the contextuality-powered tasks we discuss later, and many communication and cryptographic protocols are all naturally phrased in it. Both the generalized-probabilistic-theory description of a system and the notion of (non)contextuality we use are tailored to this setting.}

\begin{figure}[h]
\centering

{%
\begin{tikzpicture}[>={Stealth[length=2mm]}, font=\small,
  device/.style={draw, thick, rounded corners=2pt, minimum width=1.95cm, minimum height=0.9cm, align=center},
  proc/.style={circle, fill=black, inner sep=1.2pt},
  class/.style={draw, dashed, rounded corners=5pt, inner sep=3.5pt}]
  \node[device] (prep) at (0,0) {preparation};
  \node[device] (meas) at (4.3,0) {measurement};
  \draw[->, thick] (prep) -- node[above]{\footnotesize system} (meas);
  \draw[->, thick] (prep.north) ++(0,0.5) node[above]{\footnotesize $P$} -- (prep.north);
  \draw[->, thick] (meas.north) ++(0,0.5) node[above]{\footnotesize $M$} -- (meas.north);
  \draw[->, thick] (meas.east) -- ++(0.65,0) node[right]{\footnotesize $k$};
  \node[proc] (p1) at (-1.0,-1.8) {};
  \node[inner sep=0pt] (p1l) at (-1.0,-2.1) {\scriptsize $P_1$};
  \node[proc] (p2) at (-0.25,-1.8) {};
  \node[inner sep=0pt] (p2l) at (-0.25,-2.1) {\scriptsize $P_2$};
  \node[proc] (p3) at (0.85,-1.8) {};
  \node[inner sep=0pt] at (0.85,-2.1) {\scriptsize $P_3$};
  \node[class, fit=(p1)(p2)(p1l)(p2l)] (sclass) {};
  \node[anchor=south] at ($(sclass.north)+(-0.2,0)$) {\scriptsize $s=[P_1]=[P_2]$};
  \draw[->, gray, shorten >=2pt] (sclass.north east) to[bend right=25] (prep.south);
  \node[proc] (e1) at (3.35,-1.8) {};
  \node[inner sep=0pt] (e1l) at (3.35,-2.1) {\scriptsize $(k|M)$};
  \node[proc] (e2) at (4.6,-1.8) {};
  \node[inner sep=0pt] (e2l) at (4.6,-2.1) {\scriptsize $(k'|M')$};
  \node[proc] (e3) at (5.9,-1.8) {};
  \node[inner sep=0pt] at (5.9,-2.1) {\scriptsize $(k''|M'')$};
  \node[class, fit=(e1)(e2)(e1l)(e2l)] (eclass) {};
  \node[anchor=south] at ($(eclass.north)+(0.25,0)$) {\scriptsize $e=[k|M]=[k'|M']$};
  \draw[->, gray, shorten >=2pt] (eclass.north west) to[bend left=25] (meas.south);
  \node at (2.15,-2.9) {\footnotesize $p(k\,|\,M,P)=B(e,s)$};
\end{tikzpicture}}
\caption{\justifying {A prepare-and-measure scenario. A preparation procedure $P$ outputs a physical system, which is sent to a measurement procedure $M$ that returns an outcome $k$, with statistics $p(k\,|\,M,P)$. Operationally equivalent procedures (dashed regions) are grouped into equivalence classes: states $s=[P]\in\bar\Omega$ and effects $e=[k|M]\in\mathcal{E}$, in terms of which the statistics read $B(e,s)$. In most of this work we deal directly with states and effects; the underlying procedures become relevant again in Sec.~\ref{sec:Contextuality}.}}
\label{fig:PrepareAndMeasure}
\end{figure}
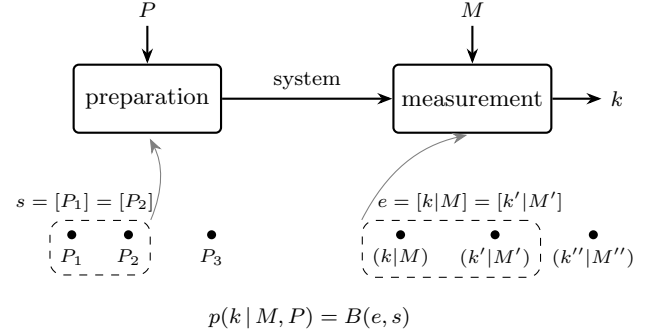

\subsection{Generalized probabilistic theories}

The framework of generalized probabilistic theories consists of a simple mathematical description that  nevertheless allows for the study of many possible physical theories (quantum and classical theories as examples). For comprehensive introductions we refer the reader to specialized works on GPTs, such as~\cite{barrett_GPTsInformationProcessing,Mueller_GPTLesHouche2021,Plavala_GPTIntro2023,Schmid_ShadowsSubsystems2025}. Here we present only a minimal background (mainly the notion of a GPT system), following the conventions of Ref.~\cite{Schmid_ShadowsSubsystems2025}, sufficient to state and understand our results.

{The GPT description of a system encodes the operational data of a prepare-and-measure scenario, focusing only on the aspects of the procedures that are relevant for the probabilities, namely, the states (equivalence classes of preparations) and effects (equivalence classes of measurement events). These are represented by vectors in real vector spaces, in terms of which the probabilities become a bilinear function. Concretely, a GPT system} specifies: (i) a compact convex set of normalized states $\bar{\Omega}$ {(each state $s$ being a vector that represents an equivalence class of preparations)}; (ii) a convex closed effect space $\mathcal{E}$ {(each effect $e=[k|M]$ a vector representing an equivalence class of measurement events)}; and (iii) a bilinear probability rule $B:\mathsf{Span}[\mathcal{E}]\times\mathsf{Span}[\bar{\Omega}]\to\mathbb{R}$ with $B(e,s)\in[0,1]$ for all $(e,s)\in\mathcal{E}\times\bar{\Omega}${, which recovers the observed statistics as $B(e,s)=p(k\,|\,M,P)$ for any representatives $P\in s$ and $(k|M)\in e$ (cf.~Fig.~\ref{fig:PrepareAndMeasure})}. Convexity of the state (effect) set captures the idea that mixtures of valid states (effects) should be valid; compactness ensures they are bounded and closed --  closedness entails that procedures we can approximate arbitrarily well should be allowed.

Among the effects, there is a distinguished one, the \emph{deterministic effect} $u\in\mathcal{E}$, which returns the probability $1$ for every normalized state: $B(u,s)=1$ for all $s\in\bar{\Omega}$. Operationally, $u$ is simply asking: \emph{is the system there?}\footnote{
It is also possible (though less common) to consider theories in which $u$ is non-unique; see, e.g.,~\cite{Chiribella_Purification2010}.
}
Additionally, for any $e\in\mathcal{E}$, the complement $u-e$ is also a valid effect ($u-e\in\mathcal{E}$), and a measurement is any finite family $\{e_i\}\subset\mathcal{E}$ with $\sum_i e_i=u$.

States and effects \emph{separate} each other with respect to $B$: for any pair of distinct states there exists an effect that distinguishes them, and vice versa. Equivalently, the bilinear form $B(\cdot,\cdot)$ is nondegenerate.\footnote{
Much of the GPT literature adopts a dual-space description, where effects live in the dual of $\mathrm{Span}[\Omega]$. Ref.~\cite{Schmid_ShadowsSubsystems2025} shows that when the separation property holds (as for full GPT systems or relatively tomographic fragments~\cite{Selby_2023}, and throughout this manuscript), the two descriptions are equivalent: each can be translated into the other.
}

Bilinearity of $B$ together with $B(u,s)=1$ immediately implies that $\bar{\Omega}$ cannot contain the zero vector: since $B$ is linear in its second slot, $B(u,0)=0$, which is incompatible with being a normalized state. One can also consider the cone
generated by the convex hull of $\bar{\Omega}$ and the $0$ vector, whose elements correspond to \emph{subnormalized
states} (though these will not be essential to our work).
Therefore, normalized states live on an affine hyperplane which does not pass through the origin. This has an important consequence for us: if $\bar{\Omega}$ has affine dimension $d$, then $\dim\!\big(\mathsf{Span}[\bar{\Omega}]\big)=d+1$, so one can always select (at most) $d\!+\!1$ linearly independent normalized states in $\bar{\Omega}$. Because all normalized states lie on this hyperplane, linear and affine (in)dependence coincide for them; throughout, we phrase everything in terms of linear (in)dependence.

As a concrete illustration, consider a qubit system. The state space $\bar{\Omega}$ is the Bloch ball (density operators with unit trace), which is three-dimensional. It is usual to represent quantum normalized states as density matrices, and the space spanned by all qubit states, $\mathsf{Span}[\bar{\Omega}]$, is the four-dimensional real vector space of $2\times2$ Hermitian operators. Therefore, one can have up to $4$ different qubit states making up a linearly independent set {(linear independence being understood among the states as vectors in $\mathsf{Span}[\bar\Omega]$, i.e.\ among the density matrices)} (for example, $\ketbra{0}{0},\ketbra{1}{1},\ketbra{+}{+},\ketbra{+y}{+y}$, where the first two are eigenstates of $\sigma_Z$ and the other two of $\sigma_X,\sigma_Y$, respectively). Effects are POVM elements with deterministic effect $u=\mathds{1}$, and note that complements $\mathds{1}-E$ are always valid POVM elements. The quantum probability rule is Born's rule $B(E,\rho)=\mathrm{Tr}[E\rho]$ and the states and effects separate each other through the Born rule: for every pair of distinct states, there exists an effect (POVM element) for which the probabilities are distinct, and vice-versa. Finally, note that indeed both $\bar{\Omega}$ and $\mathcal{E}$ are compact and convex. Other GPT systems arise just by choosing different $(\bar{\Omega},\mathcal{E},B)$ with the same structural properties. See Fig.~\ref{fig:GPTSystemsExample} for an example.

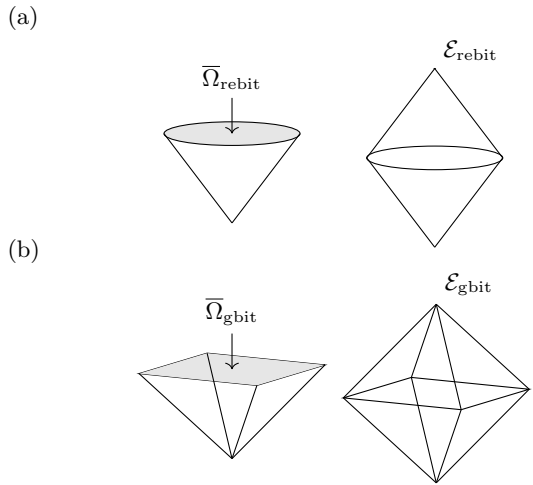
\begin{figure}[h]
\centering
\begin{subfigure}{\columnwidth}
\caption{}
\begin{minipage}{0.3\columnwidth}
\centering
\tdplotsetmaincoords{80}{180}

\begin{tikzpicture}[tdplot_main_coords, scale=0.6]

\def\r{1.5}
\def\h{2}

\coordinate (T) at (0,0,-\h);

\fill[gray!20] (0,0,0) circle (\r);
\draw (0,0,0) circle (\r);
\draw[->] (0,0,0.8) -- (0,0,0);
\draw (0,0,0.8) circle (0) node[above] {$\overline{\Omega}_{\mathrm{rebit}}$};

\draw (\r,0,0) -- (T);
\draw (-\r,0,0) -- (T);

\end{tikzpicture}
\end{minipage}
\begin{minipage}{0.3\columnwidth}
\centering
\tdplotsetmaincoords{80}{180}

\begin{tikzpicture}[tdplot_main_coords, scale=0.6]

\def\r{1.5}
\def\h{2}

\coordinate (B) at (0,0,-\h);
\coordinate (T) at (0,0,\h);

\draw (0,0,0) circle (\r);
\draw (T) circle (0.01) node[above right] {$\mathcal{E}_{\mathrm{rebit}}$};

\draw (\r,0,0) -- (T);
\draw (-\r,0,0) -- (T);
\draw (\r,0,0) -- (B);
\draw (-\r,0,0) -- (B);

\end{tikzpicture}
\end{minipage}
\end{subfigure}
\vspace*{-0.5em}
\begin{subfigure}{\columnwidth}
\caption{}
\begin{minipage}{0.3\columnwidth}
\centering
\tdplotsetmaincoords{80}{30}

\begin{tikzpicture}[tdplot_main_coords, scale=0.6]

\def\l{1.5}
\def\h{2}

\coordinate (V1) at (\l,\l,0);
\coordinate (V2) at (-\l,\l,0);
\coordinate (V3) at (-\l,-\l,0);
\coordinate (V4) at (\l,-\l,0);
\draw (V1) -- (V2) -- (V3) -- (V4) -- cycle;
\fill[gray!20] (V1) -- (V2) -- (V3) -- (V4) -- cycle;
\draw[->] (0,0,0.8) -- (0,0,0);
\draw (0,0,0.8) circle (0) node[above] {$\overline{\Omega}_{\mathrm{gbit}}$};

\coordinate (B) at (0,0,-\h);

\draw (B) -- (V1);
\draw (B) -- (V2);
\draw (B) -- (V3);
\draw (B) -- (V4);



\end{tikzpicture}
\end{minipage}
\begin{minipage}{0.3\columnwidth}
\centering
\tdplotsetmaincoords{80}{30}

\begin{tikzpicture}[tdplot_main_coords, scale=0.6]

\def\l{1.5}
\def\h{2}

\coordinate (V1) at (\l,\l,0);
\coordinate (V2) at (-\l,\l,0);
\coordinate (V3) at (-\l,-\l,0);
\coordinate (V4) at (\l,-\l,0);
\draw (V1) -- (V2) -- (V3) -- (V4) -- cycle;

\coordinate (B) at (0,0,-\h);

\draw (B) -- (V1);
\draw (B) -- (V2);
\draw (B) -- (V3);
\draw (B) -- (V4);

\coordinate (T) at (0,0,\h);

\draw (T) -- (V1);
\draw (T) -- (V2);
\draw (T) -- (V3);
\draw (T) -- (V4);

\draw (T) circle (0) node[above right] {$\mathcal{E}_{\mathrm{gbit}}$};



\end{tikzpicture}
\end{minipage}
\end{subfigure}
\caption{\justifying Two GPT systems with state and effect spaces embedded in $\mathbb{R}^3$.
(a) A GPT realization using only qubit states with real density matrices and similarly for its effects (the rebit~\cite{Wootters_RebitSharing2010}). Its normalized state space is the $XZ$ equator of the Bloch ball.
(b) The gbit~\cite{barrett_GPTsInformationProcessing}, whose normalized state space is a square and whose effects satisfy $0 \le B(e,s) \le 1$.
The gbit is strongly nonclassical: it is maximally contextual and a pair of gbits can violate CHSH up to the algebraic maximum~\cite{Schmid_ShadowsSubsystems2025, Janotta_LocalStateSpacePolygon2011}.}
\label{fig:GPTSystemsExample}
\end{figure}
\subsection{Performance in state discrimination tasks}

In SD tasks, one party prepares a GPT system according to a predefined set\footnote{Strictly speaking, an SD task uses a \emph{family} of states, ie, it allows some states to coincide. However, this is usually not a case of interest, since discriminating among these identical states is just a random guess. We therefore consider just a set, where $s_k\neq s_{k'}$ for $k\neq k'$.} of available states $\{s_k\}_k\subset\bar{\Omega}$ and an a priori distribution for the choice of each state, $\{p_k\}_k$. Another party then implements a measurement to try to guess which state was actually prepared. In the minimum-error version that we focus on, that party must always take a guess on which state was prepared.
 The crucial questions in SD tasks are: $i)$ given an ensemble $\{(p_k,s_k)\}_k$, what is the maximum probability of correctly guessing which state was prepared,
\begin{align}
    P_{\rm guess}[\{(p_k,s_k)_k\}]:= \max_{\{F_k\}}\sum_{k}p_kB(F_k,s_k);
\end{align}
$ii)$ which measurement (represented by the collection $\{F_k\}\subset\mathcal{E}$, with as much outcomes as there are states) provides such a maximum success.

Analytical solutions to both questions are known only in limited scenarios in quantum theory. For example, a geometric analytical procedure exists for ensembles of qubit states~\cite{Deconinck_QubitStateDiscrimination2010}. Furthermore, closed formulas for the optimal guessing probability $P_{\rm guess}$,  are available only in particularly simple scenarios, such as discrimination between two quantum states~\cite{Helstrom1967} or certain small ensembles of qubit states~\cite{Donghoonha_2021SDQubits}. Even fewer analytical results are known in the more general setting of GPT systems~\cite{Bae_StructureSDGPT}. In general, determining $P_{\rm guess}$ amounts to solving an optimization problem that can be formulated as a semidefinite program~\cite{EldarSDPforSD,Bae_StructureSDQuantum}. Although this formulation makes the problem efficiently solvable numerically, it also suggests that closed analytical formulas for $P_{\rm guess}$ are not often found because these are unlikely to exist in general~\cite{Deconinck_QubitStateDiscrimination2010}. In this context, general  bounds on $P_{\rm guess}$, depending on the ensemble $\{(p_k,s_k)\}_k$ that characterizes the task, are extremely valuable; for example, in the case of quantum systems, the so-called `pretty good measurement'  provides a good approximation for $P_{\rm guess}$~\cite{HausladenWootters1994,BarnumKnill2002}. Despite decades of research, and given how  prevalent SD tasks are in quantum information, new bounds are still being proposed~\cite{montanaro2019prettysimpleboundsquantum}.

Recall that the problem to which we devote is: how well does a given state set $\{s_k\}$ perform in a task? Or, similarly, given a desired performance, which conditions should hold for the state set? Applying this perspective to SD tasks,  we should consider all and only those tasks in which every state in the set has a non-null chance of being prepared -- otherwise, one would be studying the performance of a strict subset of states instead of the performance of $\{s_k\}$. Mathematically, given the set of states, we consider its performance in all SD tasks $\{(p_k,s_k)\}$ with $p_k>0$ for all $k$. An important consequence of that is that the average state $\bar{s}=\sum_kp_ks_k$ of an SD task $\{(p_k,s_k)\}$, is an element of $\mathsf{ConvHull}[\{s_k\}]\subset\mathsf{Span}[\{s_k\}]$. As a consequence, we get the following:

{
\begin{rmk}
\label{rmk:AffIndepDec}
Consider an SD task $\{(p_k,s_k)\}_k$ with $p_k>0$ for all $k$ and $\{s_k\}_k$ a linearly dependent set of states of a GPT system, with average state $\bar s:=\sum_k p_k s_k$. Since $\bar s\in\mathsf{ConvHull}[\{s_k\}]$, the average state always admits (typically several) \textbf{linearly independent decompositions}: for each linearly independent subset $S_{\rm LI}\subseteq\{s_k\}$ with $\bar s\in\mathsf{ConvHull}[S_{\rm LI}]$ there is a unique probability distribution $\{q_k\}_k$ supported on $S_{\rm LI}$ with
\begin{align}
\label{eq:AffIndDecompSbar}
    \bar s=\sum_k q_k s_k.
\end{align}
We refer to such a $\{q_k\}$ as a linearly independent decomposition of $\bar s$ from $\{s_k\}$; its support recovers $S_{\rm LI}$, so distinct decompositions have distinct supports.
Any such $\{q_k\}$ necessarily differs from the prior $\{p_k\}$: since $\{s_k\}$ is, by assumption, linearly dependent, every linearly independent subset $S_{\rm LI}$ is a \emph{proper} subset of $\{s_k\}$, so $\{q_k\}$ vanishes on at least one index, whereas $p_k>0$ for all $k$. The total variation distance between $\{q_k\}$ and $\{p_k\}$ captures this difference,
\begin{align}
\label{eq:Vdef}
    V(\{q_k\},\{p_k\}):=\tfrac12\sum_k|q_k-p_k|,
\end{align}
and is therefore strictly positive.
\end{rmk}
}
See Fig.~\ref{fig:LIDecompositions} for an example of different linearly independent decompositions for the same $\bar{s}$ and set $\{s_k\}$.

\begin{figure}[h]
\centering

\begin{subfigure}{0.45\columnwidth}
\caption{}
\centering
\begin{tikzpicture}[scale=1.5,baseline={(0,0)}]
    \draw (0,0) circle (1.0); 
    \fill (0.623, 0.782) circle (1pt) node[above right] {$s_{1}$};
    \fill (-0.222, 0.974) circle (0pt) node[above left] {$s_{2}$};
    \fill ($0.9*(-0.222, 0.974)$) circle (1pt);
    \fill (-0.900, 0.433) circle (0pt) node[above left] {$s_{3}$};
    \fill ($0.8*(-0.900, 0.433)$) circle (1pt);
    \fill (-0.900, -0.433) circle (1pt) node[below left] {$s_{4}$};
    \fill (-0.222, -0.974) circle (0pt) node[below left] {$s_{5}$};
    \fill ($0.9*(-0.222, -0.974)$) circle (1pt);
    \fill (0.623, -0.781) circle (0pt) node[below right] {$s_{6}$};
    \fill ($0.85*(0.623, -0.781)$) circle (1pt);
    \fill (1, 0) circle (0pt) node[right] {$s_{7}$};
    \fill ($0.9*(1, 0)$) circle (1pt);
    \draw (0.623, 0.782) --  ($0.9*(-0.222, 0.974)$) -- ($0.8*(-0.900, 0.433)$) -- (-0.900, -0.433) -- ($0.9*(-0.222, -0.974)$) -- ($0.85*(0.623, -0.781)$) -- ($0.9*(1, 0)$) -- cycle;
    \fill ($0.7*(-0.623, -0.781)$) circle (1pt) node[above] {$s_{8}$};
    \fill ($0.3*(0.623, 0.782)$) circle (1pt) node[above] {$\overline{s}$};
\end{tikzpicture}
\end{subfigure}
\hfill
\begin{subfigure}{0.45\columnwidth}
\caption{}
\centering
\begin{tikzpicture}[scale=1.5,baseline={(0,0)}]
    \draw (0,0) circle (1.0); 
    \fill (0.623, 0.782) circle (1pt) node[above right] {$s_{1}$};
    \fill (-0.222, 0.974) circle (0pt) node[above left] {$s_{2}$};
    \fill ($0.9*(-0.222, 0.974)$) circle (1pt);
    \fill (-0.900, 0.433) circle (0pt) node[above left] {$s_{3}$};
    \fill ($0.8*(-0.900, 0.433)$) circle (1pt);
    \fill (-0.900, -0.433) circle (1pt) node[below left] {$s_{4}$};
    \fill (-0.222, -0.974) circle (0pt) node[below left] {$s_{5}$};
    \fill ($0.9*(-0.222, -0.974)$) circle (1pt);
    \fill (0.623, -0.781) circle (0pt) node[below right] {$s_{6}$};
    \fill ($0.85*(0.623, -0.781)$) circle (1pt);
    \fill (1, 0) circle (0pt) node[right] {$s_{7}$};
    \fill ($0.9*(1, 0)$) circle (1pt);
    \fill ($0.7*(-0.623, -0.781)$) circle (1pt) node[above] {$s_{8}$};
    \fill ($0.3*(0.623, 0.782)$) circle (1pt) node[above] {$\overline{s}$};
    \draw[color = red] (0.623, 0.782) -- ($0.9*(-0.222, 0.974)$) -- ($0.85*(0.623, -0.781)$) -- cycle;
    \draw[color = blue]  ($0.9*(-0.222, 0.974)$) -- (-0.900, -0.433) -- ($0.9*(1, 0)$) -- cycle;
    \draw[color = red,->] ($1.05*(0.21, 0.829)$) -- ($1.5*(0.21, 0.829)$);
    \fill[color = red] ($1.5*(0.21, 0.829)$) circle (0pt) node[above] {$S$};
    \draw[color = blue,->] ($(0.349+0.05, 0.438)$) -- ($(0.349+0.7, 0.438)$);
    \fill[color = blue] ($(0.349+0.7, 0.438)$) circle (0pt) node[right] {$S'$};
\end{tikzpicture}
\end{subfigure}

\caption{\justifying Consider an SD task for the rebit state space, as defined in Fig~\ref{fig:GPTSystemsExample}. In (a) we depict some set of states $\{s_{k}\}_{k=1}^{8}$ along with the average state $\overline{s}$ for that task. The state  $\overline{s}$ admits linearly independent decompositions from both $S_{\rm LI} = \{s_{1},s_{2},s_{6}\}$ and $S^{'}_{\rm LI} = \{s_{2},s_{4},s_{7}\}$. }
\label{fig:LIDecompositions}

\end{figure}
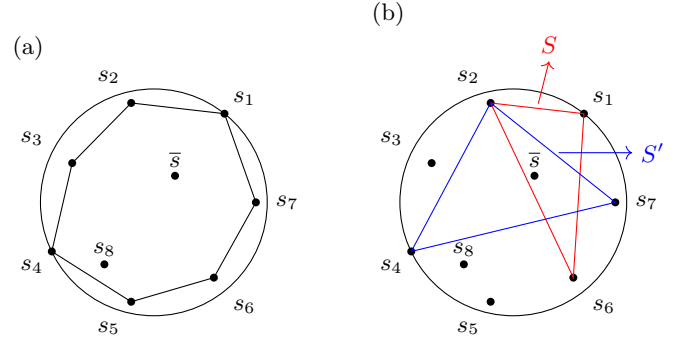

{In what follows we develop the technical results: the SD bounds and their main properties. These rest on the purely geometric notion of linear (in)dependence. The main physical significance of the results and the most important consequence of the bounds appear in Sec.~\ref{sec:Contextuality}, where the connection to contextuality is established and discussed.}

\section{Bounds for State Discrimination from Linear Independence}
\label{sec:Technical}

We now state and analyze our bounds from a geometric point of view. Throughout, we consider SD tasks that use all states in a subset $\{s_k\}\subset\Omega$ of some GPT system, with strictly positive priors $\{p_k\}$ (as we want to analyze the performance of the full set in SD tasks) and average state ${\bar s=\sum_k p_k s_k}$. {The goal of this section is to show that a single geometric feature of the set, namely, whether it is linearly dependent, already limits how well its states can be discriminated. The ideal case is immediate, ie  perfect discrimination requires the states to be linearly independent, but what we show here is the quantitative, robust version of this fact, namely that dependence forces the success probability strictly below $1$, by at least an amount we can compute from the geometry alone. We stress that at this stage linear dependence is invoked only as a convenient, testable geometric hypothesis; its physical motivation comes later, through its link to contextuality (Sec.~\ref{sec:Contextuality}).}

{
\begin{restatable}[Upper bound on $P_{\mathrm{guess}}$ from linear dependence]{theorem}{ThmLD}
\label{BoundsLinDepSD}
Let $\{s_k\}_k\subset\bar{\Omega}$ be a finite, linearly dependent set of states in a GPT system, and consider \emph{any} SD task $\{(p_k,s_k)\}$ that uses \emph{all} of them (i.e., $p_k>0$ for all $k$), with average state $\bar s:=\sum_k p_ks_k$.
For any linearly independent decomposition $\{q_k\}$ of $\bar s$ obtained from a linearly independent subset $S_{\rm LI}\subseteq\{s_k\}$ (Remark~\ref{rmk:AffIndepDec}), the optimal guessing probability obeys
\begin{align}
\label{eq:BoundsLD}
    P_{\mathrm{guess}}\ &\le\ 1-\mathsf P_{\{q_k\}}, \qquad\text{with}\\
    \mathsf P_{\{q_k\}}\ &:=\ p_{\min}\!\left(1+\frac{p_{\min}}{V(\{q_k\},\{p_k\})}\right)^{\!-1},\nonumber
\end{align}
and $p_{\min}:=\min_k p_k$.
\end{restatable}
The proof is given in Appendix~\ref{app:ProofTheorem1Corollaries1And2}.

The above upper bound holds for every linearly independent decomposition of $\bar s$ from $\{s_k\}$, that is, given a linearly independent decomposition $\{q_k\}$ of $\bar{s}$, a bound on $P_{\rm guess}$ can be calculated. Then, minimizing the bounds over all such decompositions (ie, maximizing $\mathsf P_{\{q\}}$ over $\{q_k\}$) yields a universal constraint:
\begin{restatable}{corollary}{CorMaxAndHalf}
\label{corollary:MaxAndHalfBounds}
Fix a linearly dependent set $\{s_k\}\subset\bar{\Omega}$ in some GPT system. For any SD task $\{(p_k,s_k)\}$that uses all the states in the set,
\begin{align}
P_{\mathrm{guess}} \ \le\ 1-\max_{\{q_k\}} \mathsf P_{\{q_k\}}\ \le\ 1-\tfrac12\,p_{\min},
\end{align}
where the maximum is over all linearly independent decompositions $\bar s=\sum_k q_k s_k$ obtained from $\{s_k\}$.
\end{restatable}
An immediate consequence of Corollary~\ref{corollary:MaxAndHalfBounds} is that, once an SD task is fixed (i.e., the prior distribution is given), the optimal guessing probability for any linearly dependent state set
is bounded away from unity in every GPT.

The bounds obey a useful invariance property.
\begin{restatable}[Invariance of the bounds under injective linear maps]{corollary}{BoundInvarianceInjective}
\label{corollary:invarianceBoundsINjectiveLinearM}
The bounds of Theorem~\ref{BoundsLinDepSD} and Corollary~\ref{corollary:MaxAndHalfBounds} are invariant under the action of injective linear maps on $\{s_k\}$\footnote{In the language of GPT systems in Ref.~\cite{Schmid_ShadowsSubsystems2025} , an injective linear map from the states of a GPT system G to the states of another GPT system H is called a \emph{faithful state-map}, though their definition also extends to subnormalized states.}. In other words:
let $\{s_k\}\subset\bar\Omega$ be normalized states of a GPT system and let $\{t_k\}\subset\bar\Omega'$ be normalized states of a (possibly different) GPT system with $t_k=M(s_k)$ for some injective linear map $M:\mathsf{Span}\{s_k\}\to\mathsf{Span}\{t_k\}$. Then,  the numerical bounds of Theorem~\ref{BoundsLinDepSD}, its corollaries, and Lemma~\ref{lemma:ConvexlyDependent} {(derived in Appendix~\ref{app:UniformAndConvexlyDep}; see the remark in Appendix~\ref{app:affine-invariance} for the results stated only there)} computed for SD task $\{(p_k,s_k)\}$ coincide with those computed for $\{(p_k,t_k)\}$, for every prior $(p_k)_k$ with $p_k>0$ for all $k$.
\end{restatable}

Intuitively, this invariance holds because an injective linear map preserves linear (in)dependence and the coefficients of any decomposition of the average state; since the bounds depend only on the prior $\{p_k\}$ and on the weights $\{q_k\}$ of linearly independent decompositions, they take the same value before and after the map (see Appendix~\ref{app:affine-invariance} for the proof).

Therefore, given the prior $\{p_k\}$ of an SD task, the bounds calculated for a particular set of states $\{s_k\}$   actually apply to a whole  class of sets, ie those  linearly equivalent. {This invariance is worth emphasizing: a given physical state set can be written in many mathematically distinct but operationally equivalent ways (different choices of basis, normalization, or embedding into a larger or different GPT). Corollary~\ref{corollary:invarianceBoundsINjectiveLinearM} guarantees that our bounds do not depend on any such choice---they are gauge invariant, or representation independent. They thus constrain the genuine operational content of the state set rather than an artifact of its description, and apply simultaneously to every realization of that set, even across different theories or systems.} Figure~\ref{fig:cube} depicts this fact for a particular example detailed in Appendix~\ref{app:affine-invariance}.  In Appendix~\ref{app:UniformAndConvexlyDep} we also derive a strengthened version of the bounds for the case in which some states in the set are mixtures of other states in the set, i.e., non-extremal in $\mathrm{ConvHull}[\{s_k\}]$ (Lemma~\ref{lemma:ConvexlyDependent}).

\begin{figure}[h!]
\centering

\begin{subfigure}{\columnwidth}
\caption{}
\begin{minipage}{0.48\columnwidth}
\centering
\begin{tikzpicture}[scale=1,baseline={(0,0)}]

\coordinate (A) at (0,0);
\coordinate (B) at (2,0);
\coordinate (C) at (2,2);
\coordinate (D) at (0,2);
\coordinate (A') at (-0.6,0.6);
\coordinate (B') at (1.4,0.6);
\coordinate (C') at (1.4,2.6);
\coordinate (D') at (-0.6,2.6);
\draw (A) -- (B) -- (C) -- (D) -- cycle;
\draw (A') -- (B') -- (C') -- (D') -- cycle;
\draw (A) -- (A');
\draw (B) -- (B');
\draw (C) -- (C');
\draw (D) -- (D');
\fill[color = red] (A) circle (0.04) node[below left, color = red] {$s_1$};
\fill[color = red] (B) circle (0.04) node[below right, color = red] {$s_2$};
\fill[color = red] (B') circle (0.04) node[right, color = red] {$s_3$};
\fill[color = red] (A') circle (0.04) node[below left, color = red] {$s_4$};
\fill[color = red] (D) circle (0.04) node[above, color = red] {$s_5$};
\fill[color = red] (C) circle (0.04) node[above right, color = red] {$s_6$};
\fill[color = red] (C') circle (0.04) node[above right, color = red] {$s_7$};
\fill[color = red] (D') circle (0.04) node[above left, color = red] {$s_8$};
\coordinate (Sbar) at (0.7,1.3);
\fill[purple] (Sbar) circle (0.05);
\node[purple, right] at (Sbar) {$\bar{s}$};

\draw[->] (2.3,1.4) to[out=30, in=210] (3.3,1.4);
\draw[->] (3.3,1.2) to[out=210, in=30] (2.3,1.2);

\end{tikzpicture}
\end{minipage}
\begin{minipage}{0.48\columnwidth}
\centering
\begin{tikzpicture}[scale=0.6, baseline={(-0.5,-0.5)}]

\coordinate (A) at (0,0);
\coordinate (B) at (2,0);
\coordinate (C) at (2,2);
\coordinate (D) at (0,2);
\coordinate (A') at (-0.6,0.6);
\coordinate (B') at (1.4,0.6);
\coordinate (C') at (1.4,2.6);
\coordinate (D') at (-0.6,2.6);
\draw (A) -- (B) -- (C) -- (D) -- cycle;
\draw (A') -- (B') -- (C') -- (D') -- cycle;
\draw (A) -- (A');
\draw (B) -- (B');
\draw (C) -- (C');
\draw (D) -- (D');

\fill[color = red] (A) circle (0.04) node[below left, color = red] {$t_1$};
\fill[color = red] (B) circle (0.04) node[below right, color = red] {$t_2$};
\fill[color = red] (B') circle (0.04) node[scale = 0.8,right, color = red] {$t_3$};
\fill[color = red] (A') circle (0.04) node[below left, color = red] {$t_4$};
\fill[color = red,] (D) circle (0.04) node[scale = 0.8, left, color = red] {$t_5$};
\fill[color = red] (C) circle (0.04) node[above right, color = red] {$t_6$};
\fill[color = red] (C') circle (0.04) node[above right, color = red] {$t_7$};
\fill[color = red] (D') circle (0.04) node[above left, color = red] {$t_8$};
\fill[color = red] (A) circle (0.08);
\fill[color = red] (B) circle (0.08);
\fill[color = red] (B') circle (0.08);
\fill[color = red] (A') circle (0.08);
\fill[color = red] (D) circle (0.08);
\fill[color = red] (C) circle (0.08);
\fill[color = red] (C') circle (0.08);
\fill[color = red] (D') circle (0.08);
\coordinate (Sbar) at (0.7,1.3);
\fill[purple] (Sbar) circle (0.08);
\node[purple, right] at (Sbar) {$\bar{t}$};

\draw[dashed] (Sbar) circle (1.83);
\draw[dashed] (Sbar) ellipse (1.83 and 0.4);

\end{tikzpicture}
\end{minipage}
\end{subfigure}

\vspace{1em}

\begin{subfigure}{\columnwidth}
\caption{}
\begin{minipage}{0.48\columnwidth}
\centering
\begin{tikzpicture}[scale=1,baseline={(0,0)}]

\coordinate (A) at (0,0);
\coordinate (B) at (2,0);
\coordinate (C) at (2,2);
\coordinate (D) at (0,2);
\coordinate (A') at (-0.6,0.6);
\coordinate (B') at (1.4,0.6);
\coordinate (C') at (1.4,2.6);
\coordinate (D') at (-0.6,2.6);
\draw (A) -- (B) -- (C) -- (D) -- cycle;
\draw (A') -- (B') -- (C') -- (D') -- cycle;
\draw (A) -- (A');
\draw (B) -- (B');
\draw (C) -- (C');
\draw (D) -- (D');
\fill[color = red] (A) circle (0.04) node[below left, color = red] {$s_1$};
\fill[color = red] (B) circle (0.04) node[below right, color = red] {$s_2$};
\fill[color = red] (B') circle (0.04) node[right, color = red] {$s_3$};
\fill[color = red] (A') circle (0.04) node[below left, color = red] {$s_4$};
\fill[color = red] (D) circle (0.04) node[above, color = red] {$s_5$};
\fill[color = red] (C) circle (0.04) node[above right, color = red] {$s_6$};
\fill[color = red] (C') circle (0.04) node[above right, color = red] {$s_7$};
\fill[color = red] (D') circle (0.04) node[above left, color = red] {$s_8$};
\coordinate (Sbar) at (0.7,1.3);
\fill[purple] (Sbar) circle (0.05);
\node[purple, right] at (Sbar) {$\bar{s}$};

\draw[->] (2.3,1.4) to[out=30, in=210] (3.3,1.4);
\draw[->] (3.3,1.2) to[out=210, in=30] (2.3,1.2);

\end{tikzpicture}
\end{minipage}
\begin{minipage}{0.48\columnwidth}
\centering
\begin{tikzpicture}[xscale=1,yscale=0.5, baseline={(-1,-0.75)}]


\begin{scope}[rotate around={20:(0.7,1.3)}]

\coordinate (A) at (0,0);
\coordinate (B) at (2,0);
\coordinate (C) at (2,2);
\coordinate (D) at (0,2);

\coordinate (A') at (-0.6,0.6);
\coordinate (B') at (1.4,0.6);
\coordinate (C') at (1.4,2.6);
\coordinate (D') at (-0.6,2.6);

\draw (A) -- (B) -- (C) -- (D) -- cycle;
\draw (A') -- (B') -- (C') -- (D') -- cycle;

\draw (A) -- (A');
\draw (B) -- (B');
\draw (C) -- (C');
\draw (D) -- (D');

\fill[color = red] (A) circle (0.06) node[below left, color = red] {$t_1$};
\fill[color = red] (B) circle (0.06) node[below right, color = red] {$t_2$};
\fill[color = red] (B') circle (0.06) node[above right, color = red] {$t_3$};
\fill[color = red] (A') circle (0.06) node[below left, color = red] {$t_4$};
\fill[color = red] (D) circle (0.06) node[scale = 0.9,above right, color = red] {$t_5$};
\fill[color = red] (C) circle (0.06) node[above right, color = red] {$t_6$};
\fill[color = red] (C') circle (0.06) node[above right, color = red] {$t_7$};
\fill[color = red] (D') circle (0.06) node[scale = 1,above, color = red] {$t_8$};

\foreach \p in {A,B,C,D,A',B',C',D'} {
  \fill[red] (\p) circle (0.04);
}

\end{scope}

\coordinate (Sbar) at (0.7,1.3);
\fill[purple] (Sbar) circle (0.05);
\node[purple, right] at (Sbar) {$\bar{t}$};

\end{tikzpicture}
\end{minipage}
\end{subfigure}

\caption{\justifying Illustration of Corollary~\ref{corollary:invarianceBoundsINjectiveLinearM}. An injective linear map sends the cube configuration $\{s_k\}$ with average $\bar{s}$ to a new set $\{t_k\}$ with average $\bar{t}$. (a) Uniform scaling of the cube that fits it into the state space of a qubit. (b) Diagonal stretching. For uniform priors the same SD bound holds in both cases, $P_{\mathrm{guess}}\le 25/28$. See Appendix~\ref{app:affine-invariance}. }
\label{fig:cube}
\end{figure}
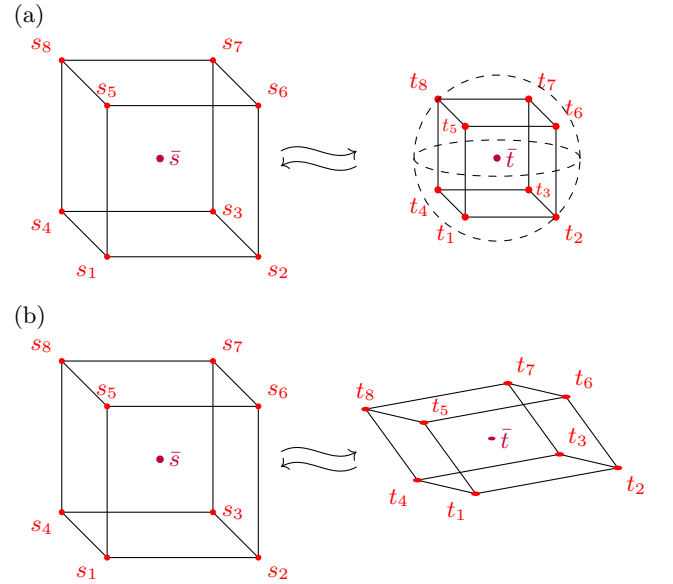

Let us now analyze the bounds. The usefulness of a bound on $P_{\rm guess}$ can be assessed in several ways; for instance, by how closely it approximates the optimal value, how easy it is to compute, or how broadly it applies. The strength of our bounds lies not in their tightness, but in \emph{their universality and geometric character}. Because they depend only on the convex geometry of the state set and the prior distribution, they provide theory-independent limits that can be evaluated directly from simple geometric features of $\{s_k\}$.  In particular, the bounds remain valid in non-quantum GPTs where standard quantum state-discrimination techniques, such as Barnum--Knill or pretty-good-measurement estimates~\cite{HausladenWootters1994,BarnumKnill2002}, may not be available. In many situations the bounds yield non-trivial analytical constraints on discrimination tasks without solving the optimization problem defining $P_{\rm guess}$, which can be of practical advantage.

Because of this geometric and universal character, the bounds can serve as feasibility tests when choosing state sets for discrimination tasks, again without having to calculate $P_{\rm guess}$ itself. Suppose a protocol requires a success probability $P^*$ for a given number of labels and prior distribution. If for a candidate state set $\{s_k\}$ one finds a linearly independent decomposition $\{q_k\}$ such that $P^*>1-\mathsf{P}_{\{q_k\}}$, then no measurement strategy can achieve the desired performance with that set (nor with any linearly equivalent one, by Corollary~\ref{corollary:invarianceBoundsINjectiveLinearM}). Thus, Theorem~\ref{BoundsLinDepSD} allows one to discard the candidate set $\{s_k\}$ and any linearly equivalent one.

If the prior is uniform, we show in the Appendix (see Corollary~\ref{corollary: UniformPrior}) that having $P^*>1-\frac{1}{N}\!\left(\frac{N-D}{N-D+1}\right)$ for some integer $D$ implies that no set of states of any GPT system satisfying $\dim(\mathrm{Span}[\{s_k\}])=D$ can achieve the task. Finally, using only the prior distribution, whatever it is, Corollary~\ref{corollary:MaxAndHalfBounds} tells us that if $P^*>1-\frac{p_{\rm min}}{2}$, no set of linearly dependent states of any GPT system can achieve $P_{\rm guess}=P^*$.

As an illustration, if one is interested in a state set with discriminability $P^*=0.9$, Example~\ref{example:cube} in the Appendix tells us that eight states forming the vertices of a cube are not suitable (or any linearly equivalent set). If one requires $P^*=0.91$ for an SD task with uniform priors, then no state set with $8$ states and spanning a dimension $D=4$ can achieve the task, regardless of its geometry or of the GPT system in which it is realized. Finally, if $P^*>0.9375$, then only state sets $\{s_k\}_{k=1}^8$ in which all eight states are linearly independent can achieve such discriminability for uniform priors. Note that, even though these latter bounds are looser, they lead to stronger conclusions.

{Above, the bounds hinged on a single hypothesis: that the state set is linearly dependent. Linear dependence is a meaningful geometric feature and is numerically testable but, on its own, it carries no obvious physical content. As we show next, however, there is an important class of tasks that require linear dependence, namely the class of \emph{contextuality-powered tasks}, endowing our bounds with relevant physical implications and providing a new link between nonclassicality and discriminability.}

\section{Advantage on contextuality-powered tasks implies bounded performance on SD tasks}
\label{sec:Contextuality}

Generalized noncontextuality provides a gold standard when it comes to defining whether a theory, system or experiment admits of a classical explanation~\cite{Spekkens_Contextuality2005}\footnote{The main idea behind generalized noncontextuality resorts to the application of a Leibnizian principle, carrying a robust and deep conceptual foundation~\cite{spekkens2019_Liebniz}, which is outside the scope of this work}.
{Besides its foundational significance, contextuality has been identified as a resource behind advantages in a range of information-processing tasks, with concrete implications for quantum technologies~\cite{Howard2014,Raussendorf2013,Bermejo-Vega2017,Schmid_CAdvantageousSD}. We briefly recall the idea, specialized to the prepare-and-measure setting of this work; for the general formulation see Refs.~\cite{Spekkens_Contextuality2005,Schmid_2021_NCinGPTsystems}. Recall that experiments can be phrased in terms of the procedures and equivalence classes introduced in Sec.~\ref{sec:Preliminaries}. Faced with the operational statistics $p(k\,|\,M,P)$ of a prepare-and-measure scenario, one may try to explain them with an \emph{ontological model}. Such a model posits a set $\Lambda$ of \emph{ontic states}, each $\lambda\in\Lambda$ being a complete specification of the physical properties of the system -- the ``real state of affairs'' underlying a run of the experiment. A preparation procedure $P$ need not pin the ontic state down, only set it up with some probability, so it is represented by a probability distribution $\mu_P(\lambda)$ over $\Lambda$; likewise, the ontic state need not determine measurement outcomes with certainty, so each measurement event $(k|M)$ is represented by a response function $\xi_{k|M}(\lambda)\in[0,1]$, the probability of obtaining outcome $k$ in measurement $M$ when the ontic state is $\lambda$. The model explains the experiment when it recovers the observed statistics: $p(k|M,P)=\sum_{\lambda\in\Lambda}\mu_P(\lambda)\,\xi_{k|M}(\lambda)$. Stated this way, the scheme is no constraint at all as any experiment with any statistics admit \emph{some} ontological model. The substantive question is whether the statistics admit a model that draws no distinctions the experiment itself does not: one in which the representation of each procedure depends only on its equivalence class. A model is \emph{noncontextual} when this is the case, i.e., $P\simeq P'$ implies $\mu_P=\mu_{P'}$ and $(k|M)\simeq(k'|M')$ implies $\xi_{k|M}=\xi_{k'|M'}$. The distributions and response functions of a noncontextual model are then well-defined functions of the GPT states $s\leftrightarrow[P]$ and effects $e\leftrightarrow[k|M]$ alone, and the model takes the form
\begin{align}
\label{eq:OntModel}
B(e,s)=\sum_{\lambda\in\Lambda}\mu_s(\lambda)\,\xi_e(\lambda).
\end{align}
A system is called \emph{noncontextual or classical} when a single such model covers every experiment one can design with it; conversely, the system is called \emph{contextual} when some of its statistics admit no noncontextual model, i.e., when any ontological explanation is forced to distinguish between operationally equivalent procedures. Contextuality in this sense is a hallmark of nonclassicality and the resource behind the advantages mentioned above.}
Generalized noncontextuality subsumes many notions of classicality in the literature, such as the existence of a positive quasi-probabilistic  representation~\cite{Spekkens_NegativityAndC2008,Schmid_structureTheorem_2024,Schmid_KD_2024}, classicality emergent within quantum Darwinism processes~\cite{Baldijao_2022} and the classicality in GPTs, traditionally linked to strictly simplicial theories~\cite{barrett_GPTsInformationProcessing,Schmid_2021_NCinGPTsystems}. Furthermore, nonclassicality signaled by  generalized contextuality has larger applicability than Bell-nonlocality~\cite{Bell1964} and Kochen-Specker contextuality~\cite{KochenSpecker1967} and allows for many advantageous applications, from thermodynamics~\cite{Naim_CinAnomaloushEATfLOW,Lostaglio_CandTheromLinearResponse2020} to  cryptography~\cite{Spekkens_POM} and computation~\cite{Howard2014,Raussendorf2013}.

The characterization of noncontextuality for GPT single-systems in prepare-and-measure scenarios  was described in Ref.~\cite{Schmid_2021_NCinGPTsystems}, which showed that a GPT system $(\bar\Omega,\mathcal{E},B(\cdot,\cdot))$ is classical if and only if it can be embedded into a strictly classical GPT. Our focus, however, is in a particular consequence of such characterization that tells when a single system's set of states always admit of a classical (ie, noncontextual) explanation.

{Intuitively, if the states $\{s_k\}$ are linearly independent, then any preparation built from them, $s=\sum_k q_k\,s_k$, is uniquely determined by its weights $\{q_k\}$. One may then take the labels $k$ themselves as ontic states, represent each preparation by the distribution $\{q_k\}$ over them, and let each effect $e$ respond to ontic state $k$ with probability $\xi_e(k)=B(e,s_k)$; bilinearity of $B$ then reproduces the observed statistics, as in Eq.~\eqref{eq:OntModel}. Because the weights $\{q_k\}$ are uniquely fixed by the prepared state, the assignment depends only on the equivalence class of each preparation procedure, so the model is noncontextual. This is exactly the mechanism identified, for quantum systems, in the study of emergent classicality under quantum Darwinism~\cite{Baldijao_NCunderQD2021}, where sufficiently distinguishable environmental encodings are affinely (hence linearly) independent and therefore admit a noncontextual model. The following lemma states the GPT generalization.}

\begin{restatable}{lemma}{LinearIndepSufficientForNC}
\label{lemma:CandLinearDep}
    If a set of states $\{s_k\}\subset\bar{\Omega}$ of a GPT system is linearly independent, then it cannot allow for contextuality in prepare-and-measure experiments.
\end{restatable}
This result was proven for the special case of quantum systems in Ref.~\cite{zhang2026reassessingboundary}  and we generalize it to GPTs in the Appendix~\ref{app:ProofLemmaCandLD}.

It is thus clear that the conjunction of Theorem~\ref{BoundsLinDepSD} and Lemma~\ref{lemma:CandLinearDep} imply a link between the performance of a set of states in SD tasks and tasks in which the advantage comes from generalized contextuality. We formalize this link via the following Theorem, which constitutes our second main result.

\begin{theorem}[Advantage in contextuality powered tasks implies bounds on the success in any SD tasks]
\label{theorem:linkCandSDtasks}
Consider a set of states $\{s_k\}_k\subset \Omega$ of a GPT system. If there exists  \textbf{some} contextually-powered tasks for which $\{s_k\}$ provides advantage, then for \textbf{any} SD task that uses the full set of states ($\{(p_k,s_k)\}$ with $p_k>0$ for all $k$),  $P_{\rm guess}$ obeys the bounds of Theorem~\ref{BoundsLinDepSD}.
Contra-positively,  if there exists \textbf{some} SD task (using the full set $\{s_k\}$) in which $P_{\rm guess}$ violates some bound in Theorem~\ref{BoundsLinDepSD} (or its corollaries), then the set cannot yield an advantage in \textbf{any} contextually-powered tasks.
\end{theorem}
\begin{proof}
   {The proof follows directly from the two results just assembled. Indeed, by (the contrapositive of) Lemma~\ref{lemma:CandLinearDep}, a set $\{s_k\}$ that provides advantage in some contextuality-powered task cannot be linearly independent, hence it is linearly dependent. Theorem~\ref{BoundsLinDepSD} then applies and bounds $P_{\rm guess}$ in every SD task using the full set, which is the first statement; the second is its contrapositive.}
\end{proof}

Theorem~\ref{theorem:linkCandSDtasks} provides
a new connection between contextuality and discriminability in GPT.
To illustrate how this link operates in concrete scenarios, we consider two cases. First, if the set of states $\{s_k\}$ is well-suited for encoding the classical information represented by the labels $\{k\}_k$, in the sense that $P_{\rm guess}>1-\frac{p_{\rm min}}{2}$ for any prior $p_k$, then $\{s_k\}$ cannot provide contextual advantage. Such cases  naturally arise, for instance, in the scope of quantum Darwinism, see e.g. \cite{Zurek_QD2009,BPH2015, HKH_QuantumOriginsObjectivity2015,Korbicz_RoadsToObjectivity2021},  a process in which the interaction between a central system and many  environmental subsystems leads to the emergence of classicality. When such a process occurs, some classical information about the system is very well encoded in  portions of the environment. Ref~\cite{Baldijao_NCunderQD2021} showed that in such a case, noncontextuality emerges. Theorem~\ref{theorem:linkCandSDtasks} provides a related but different perspective: the set of `environmental' states $\{s_k\}$ that encode the classical information after a quantum Darwinism process has occurred cannot provide advantage in any contextuality driven tasks.

As a second example of this link between contextuality and discriminability, consider the case of generalized parity oblivious multiplexing tasks (gOM taks), defined below.
\begin{example}[d-ary oblivious multiplexing task]
\label{example:gOM}
In the \(d\)-ary generalization of parity-oblivious multiplexing,
Alice receives a random string $x=(x_1,\ldots,x_n)\in\{0,\ldots,d-1\}^n$ and encodes it
into a state $s_x$. The encoding is required to be oblivious to certain linear
functions of the input. For instance, one may require that for every nonzero
$a\in\mathbb{Z}_d^n$, the value of $a\cdot x \bmod d$ is hidden, meaning that the
average encoded state is the same for all attainable values of this quantity:
\begin{align}
\label{eq:obliviousness}
\sum_{x:\,a\cdot x=r} s_x
=
\sum_{x:\,a\cdot x=r'} s_x
\qquad \forall r,r'\in\mathbb{Z}_d.
\end{align}

More general oblivious multiplexing tasks impose related, possibly weaker,
obliviousness constraints.

Bob receives the system and an index \(y\in\{1,\ldots,n\}\), chosen uniformly and independently of \(x\).
His goal is to output \(c_y\in\{0,\ldots,d-1\}\) as a guess for the \(y\)-th symbol \(x_y\) by measuring the received system.
The average success probability is
\begin{align}
P^{\mathrm{OC}}_{\mathrm{succ}}
= \frac{1}{nd^n}\sum_{y=1}^n \sum_{x} p(c_y=x_y\mid x,y).
\end{align}
This task reduces to the usual binary parity-oblivious multiplexing{~\cite{Spekkens_POM,Chailloux_2016}} when \(d=2\) and the hidden combinations correspond to mod-2 parities of the input bits and is a special case of oblivious communication tasks studied in Ref.~\cite{SahaAnubhav_CinCommunication2019}.
\end{example}

Note that the probability of success in gOM, $P^{OC}_{\rm succ}$, is an average of state discrimination tasks, one for each $y$. In other words,
For each query $y$, Bob’s subtask is a state-discrimination problem with $d$ states and random prior,
\begin{align}
\Big\{\tfrac{1}{d},\, s_{x_y}\Big\}_{x_y=0}^{d-1},
\qquad
s_{x_y=v}
   = \frac{1}{d^{\,n-1}}\!\!\sum_{x:\,x_y=v}\! s_x ,
\end{align}
having an optimum probability of guessing $P^{(y)}_{\rm guess}:= \max_{\{F_v\}}B(F_v,s_{x_y=v})$.
The total payoff \(P^{\mathrm{OC}}_{\mathrm{succ}}\) is the average of the \(n\) guessing probabilities \(P_{\mathrm{guess},y}\) associated with these ensembles.
\begin{align}
   P^{\mathrm{OC}}_{\mathrm{succ}} = \frac{1}{n}\sum_{y=1}^n P^{(y)}_{\mathrm{guess}}.
\end{align}
{Any set of states $\{s_x\}$ providing nonclassical advantage in this task obeys the bounds of Theorem~\ref{BoundsLinDepSD} in every SD task that uses all states in the set. This can be seen in two ways. Directly: any encoding satisfying the obliviousness constraints, advantageous or not, is linearly dependent (each constraint in Eq.~\eqref{eq:obliviousness} is itself a nontrivial linear relation among the states), and advantageous sets are particular instances. Or through Theorem~\ref{theorem:linkCandSDtasks}: advantage in this task requires contextuality~\cite{SahaAnubhav_CinCommunication2019}, which forces linear dependence of advantageous sets in any task, with no task-specific check. In particular, there is no instance of an SD task for which $P_{\rm guess}>1-\frac{p_{\rm min}}{2}$ for such sets.}

{Theorem~\ref{theorem:linkCandSDtasks}, however, give us a more fine-grained analysis of oblivious communication tasks}. Indeed, since $P^{\rm OC}_{\rm succ}$ is an average of SD tasks, one for each $y$, if the overall probability of success $P^{\rm OC}_{\rm succ}$ is `too high', at least one of the `subtasks' requires {$P_{{\rm guess},y}$} to be above the bounds of Theorem~\ref{BoundsLinDepSD}{, and thus} the associated sub-ensemble $\{s_{x_y}\}_{x_y=0}^{d-1}$ cannot provide contextual advantage in any task.

\begin{restatable}[High $P_{\rm succ}^{OC}$ implies a classical sub-ensemble]{corollary}{CorGOM}
\label{Corollary:gOM}
Let $\{s_x\}$, with $x=(x_1,\ldots,x_n)$ and $x_i\in\{0,\ldots,d-1\}$, be the set of states used in a $d$-ary oblivious multiplexing task.
If this set provides nonclassical advantage in such cases, then for all SD tasks using all states in $\{s_x\}$, it holds that
\begin{align}
P_{\mathrm{guess}} < 1-\frac{1}{d^n}\!\left(\!\frac{d^n-D}{d^n-D+1}\!\right)<1-\frac{1}{2d^n},
\end{align}
 where $D:=\mathsf{span}[\{s_x\}_x]$.

Moreover, if the success probability in the oblivious multiplexing task satisfies
\begin{align}
P^{\mathrm{OC}}_{\mathrm{succ}} > 1-\frac{1}{2d},
\end{align}
then for at least one value of $y^*\in\{1,\ldots,n\}$, the sub-ensemble
\begin{align}
    \Big\{\tfrac{1}{d},\, s_{x_{y^*}}\Big\}_{x_{y^*}=0}^{d-1}
\end{align}
cannot provide advantage in any contextuality-powered task.
\end{restatable}
The proof can be found in Appendix~\ref{app:ProofCorollaryGOM}

These examples show that the novel link between contextuality and discriminability is not only interesting on its own and from a foundational perspective, but has also practical implications: it can be used to derive conclusions about the performance of a set of states in SD tasks given that they provide contextual advantage or the (im)possibility of such advantage given their discriminability. Moreover, as Corollary~\ref{Corollary:gOM} shows for generalized obliviousness tasks, Theorem~\ref{theorem:linkCandSDtasks} can provide constraints on the state set depending on the role that contextuality and SD might have in the particular task: too high overall probability of success $P^{OC}_{\rm guess}$ requires nonclassicality of the full ensemble together with classicality of some `sub-ensembles'. We note that this analysis is general: Theorem~\ref{theorem:linkCandSDtasks} can be used to derive analogous conclusions for any tasks that require contextuality or high discriminability, thus being applicable to any such task to be developed in the future and any GPT system.

\bigskip

\section{Concluding remarks}

The geometry of a set of states fundamentally limits how well those states can be discriminated. In this work we showed that, in any generalized probabilistic theory (GPT), if a set of states $\{s_k\}\subset\bar{\Omega}$ is linearly dependent{, as any state set allowing for contextual advantage,} then the success probability of any state discrimination (SD) task using all states in the set is necessarily bounded away from unity. Thus linear dependence of a state set directly constrains its achievable performance in SD tasks.

The bounds we derived depend only on simple geometric features of the state set and on the prior distribution of the SD task, and admit simplified forms that depend only on the prior distribution (Corollary~\ref{corollary:MaxAndHalfBounds}) or on the number of states and the dimension of their span (see Corollary~\ref{corollary: UniformPrior} in the Appendix). When some of the states in the set are mixtures of other states in the set (i.e., non-extremal in $\mathrm{ConvHull}[\{s_k\}]$), the bounds become even stronger, cf. Lemma~\ref{lemma:ConvexlyDependent} in Appendix~\ref{app:UniformAndConvexlyDep}. Because these results hold for arbitrary GPTs, they apply even in scenarios where the usual quantum tools for state discrimination are unavailable. Although the bounds are generally not tight approximations to $P_{\rm guess}$, their analytical simplicity and generality make them useful constraints in settings where exact optimization or quantum-specific methods (such as the pretty good measurement scheme) are unavailable or impractical. This usefulness is illustrated, for instance, by the generalized oblivious communication tasks of Example~\ref{example:gOM}.

Our second main result establishes a connection between the performance of a set of states in SD tasks and the states' power to exhibit generalized contextuality. Specifically, we showed that any set of states that provides a nonclassical advantage in contextuality-powered tasks must obey the bounds derived here in every SD task involving that set. Consequently, contextual advantage necessarily implies limited discriminability of the underlying states, as determined by the bounds of Theorem~\ref{BoundsLinDepSD}. Conversely, if a set of states allows sufficiently high success probability in some SD task, then it cannot provide an advantage in any contextuality-powered task: any prepare-and-measure experiment using only this set of states will admit of a noncontextual ontological model. This reveals a fundamental trade-off between discriminability and contextuality as operational resources.

This observation has implications for several information-processing scenarios. In particular, tasks or processes that require highly distinguishable states—such as the encoding of information in environmental fragments in quantum Darwinism—are therefore better performed by sets of states that are not useful resources for contextuality-powered tasks. As another illustration, in generalized oblivious communication tasks we showed that excessively high success probability implies that the full codebook must obey our SD bounds and that at least one sub-ensemble cannot exhibit contextual advantage.

It is worth contrasting these results with previous works that relate contextuality and state discrimination~\cite{Schmid_CAdvantageousSD} (which include the cases of unbiased discrimination or maximal confidence discrimination~\cite{FlattJonwooBae_CandMaximumConfidenceD,flatt2025unifyingcontextualadvantagesstate}, besides the minimum-error case discussed here). In those approaches, contextuality is inferred from violations of bounds on $P_{\rm guess}$ under specific operational constraints. In particular, for certain pairs of states and measurement scenarios, violating the bound implies that the underlying GPT system cannot admit a noncontextual model reproducing those constraints, and therefore that the \emph{entire GPT system (that is, the full set of states and effects)} is contextual. However, such results do not determine whether the particular set of states used in the SD task itself can provide contextual advantage. In contrast, the perspective adopted here addresses a different question: given a fixed set of states $\{s_k\}$, can it provide an advantage for a particular task? Our results provide general constraints on this question by relating the performance of a given state set across all SD tasks and all contextuality-powered tasks, without imposing additional operational assumptions. In this sense, the two approaches offer complementary insights into the relationship between state discrimination and contextuality.

\section*{Acknowledgements}
 We acknowledge Marco Erba and David Schmid for useful discussions. R.D.B was supported by the Perimeter Institute for Theoretical Physics and acknowledges support by the Digital Horizon Europe project FoQaCiA, Foundations of quantum computational advantage, GA No.~101070558, funded by the European Union, NSERC (Canada), and UKRI (UK). Research at Perimeter Institute is supported
in part by the Government of Canada through the
Department of Innovation, Science and Economic
Development and by the Province of Ontario
through the Ministry of Colleges and Universities. This study was financed in part by the Coordenação de Aperfeiçoamento de Pessoal de Nível Superior - Brasil (CAPES) - Finance Code 001. This study was financed, in part, by the São Paulo Research Foundation (FAPESP), Brasil. Process Numbers 2023/02986-2 and 2024/08701-2. JKK acknowledges the support of  National Science Centre (NCN) through the QuantEra
project Qucabose 2023/05/Y/ST2/00139.

\bibliographystyle{unsrtnat}
\bibliography{bib}

\appendix

\section{Proofs and More Specific Bounds}

\subsection{Proofs of Theorem~\ref{BoundsLinDepSD} and Corollary~\ref{corollary:MaxAndHalfBounds}}
\label{app:ProofTheorem1Corollaries1And2}

We will start by proving some useful Lemmas, using the notation
\(\ubar{\lambda}:=\min_k[B(e_k,s_k)]\) and \(p_{\rm min}:=\min_k [p_k]\),
where \(\{(p_k,s_k)\}\) is a given SD task and \(e_k\) denotes the optimal measurement effect for that task.
Moreover, we denote by \(V(\{q_k\},\{p_k\})\) the total variation distance between probability distributions \(\{q_k\}\) and \(\{p_k\}\):
\begin{align}
    V(\{q_k\},\{p_k\}) := \frac{1}{2}\sum_k|q_k-p_k| = \sum_{k\in Q_+} (q_k-p_k),
\end{align}
where \(Q_+:=\{k\,|\,q_k\geq p_k\}\).

\begin{lemma}
\label{lemma: SDBoundGeneral}
In every SD problem defined by an ensemble \(\{(p_k,s_k)\}\), the probability of error obeys the lower bound
\begin{align}                  \label{ineq:SDBoundGeneral}
    P_{\rm err}\geq p_{\rm min}[1-\ubar{\lambda}].
\end{align}
\end{lemma}

\begin{proof}
Denote by \(b\) a value of \(k\) such that \(B(e_b,s_b)=\ubar{\lambda}\). Then,
\begin{align}
    P_{\rm guess}&:= \sum_k p_k B(e_k,s_k)
    =\sum_{k\neq b}p_k B(e_k,s_k) + p_b B(e_b,s_b)\nonumber \\[3pt]
    &\leq (1- p_b)+p_b\ubar{\lambda} = 1-p_b[1-\ubar{\lambda}]\nonumber \\[3pt]
    &\leq 1-p_{\min}[1-\ubar{\lambda}].
\end{align}
Since \(P_{\rm guess}+P_{\rm error}=1\), this implies that
\begin{align}
    P_{\rm error}\geq p_{\min}[1-\ubar{\lambda}].
\end{align}
\end{proof}

\begin{lemma}
\label{SDBoundAffinelyDependent2}
Consider an SD problem defined by an ensemble \(\{(p_k,s_k)\}\) in which \(\{s_k\}\) is a linearly \emph{dependent} set of GPT states.
Then, for every associated simplicial SD task \(\{(q_k,s_k)\}\), the probability of error obeys
\begin{align}
\label{ineq:SDBoundAffinelyDependent}
    P_{\rm error}\geq \ubar{\lambda}V(\{q_k\},\{p_k\}).
\end{align}
Therefore,
\begin{align}
    P_{\rm error}\geq \ubar{\lambda}\max_{\{q_k\}} V(\{q_k\},\{p_k\}),
\end{align}
where the maximization is taken over all simplicial SD tasks.
\end{lemma}

\begin{proof}
\begin{align}
    P_{\rm error}:=\sum_k\sum_{k'\neq k }
    &= \sum_k\overbrace{B(e_k,\bar{s}-p_k s_k)}^{\geq 0}
    \geq \sum_{k\in S} B(e_k,\bar{s}-p_k s_k),
\end{align}
where \(S\) is any subset \(S\subseteq\{k\}_k\), a consequence of each term being non-negative.
Now, given any \(\{q_k\}\) providing a linearly independent decomposition of $\bar{s}$ from $\{s_k\}_k$, we define the subset
\(Q_+:=\{k\,|\,q_k\geq p_k\}\), which obeys \(Q_+\subset \{k\}_k\). Therefore,
\begin{align}
    P_{\rm error}&\geq \sum_{k\in Q_+} B(e_k,\bar{s}-p_k s_k)
    =\sum_{k\in Q_+} B\!\left(e_k,\sum_i q_i s_i - p_k s_k\right)\nonumber \\[3pt]
    & = \sum_{k\in Q_+}(q_k-p_k)B(e_k,s_k)
      + \overbrace{\sum_{k\in Q_+}\sum_{i\neq k} B(e_k,q_i s_i)}^{\geq 0}
      \label{boundwithsum} \\[3pt]
    &\geq \sum_{k\in Q_+}(q_k-p_k)B(e_k,s_k)
      \geq \ubar{\lambda}\sum_{k\in Q_+}(q_k-p_k).
\end{align}
In the first equality we used the fact that $\{(q_k,s_k)\}$ is a linearly independent decomposition for the task, thus
$\bar{s}=\sum_k q_k s_k$.
Since $\sum_{k\in Q_+}(q_k-p_k)=V(\{q_k\},\{p_k\})$, we obtain the result.
\end{proof}

With these results, we prove Theorem~\ref{BoundsLinDepSD}.
\ThmLD*

\begin{proof}
This follows directly from Lemma~\ref{lemma: SDBoundGeneral} and Lemma~\ref{ineq:SDBoundAffinelyDependent}.
Indeed, recalling that \(p_{\rm min}>0\), Expression~\ref{ineq:SDBoundGeneral} implies
\begin{align}
    \ubar{\lambda}\geq 1- \frac{P_{\rm error}}{p_{\rm min}}.
\end{align}
Substituting this into Expression~\ref{ineq:SDBoundAffinelyDependent} and isolating \(P_{\rm error}\) we get the result.
\end{proof}

\CorMaxAndHalf*
\begin{proof}
    The first part of the inequality immediately follows from Theorem~\ref{BoundsLinDepSD}, since for \emph{every} $\{q_{i}\}$ providing a linear decomposition of $\overline{s}$, we have
    \begin{equation}
         P_{\mathrm{guess}}\ \le\ 1-\mathsf P_{\{q_i\}} \, ,
    \end{equation}
    and therefore
    \begin{equation}
         P_{\mathrm{guess}}\ \le\ \min_{\{q_{i}\}} \left(  1-\mathsf P_{\{q_i\}} \,\right) = 1 - \max_{\{q_{i}\}}\mathsf P_{\{q_{i}\}} \, .
    \end{equation}
    To prove the second part of the inequality, we first note that since $\{s_{k}\}$ is linearly dependent, there exists a linearly independent decomposition of $\overline{s}$ such that $q_{k}=0$ for at least some value of $k$. Call $k^{*}$ such a value, then
    \begin{align}
        V(p,q) &= \sum_{k \,  | \, p_{k} \geq q_{k} } \left( p_{k} - q_{k} \right) \\
        &= \left( p_{k^{*}} - \overbrace{q_{k^{*}}}^{=0} \right) + \overbrace{\sum_{k \,  | \, p_{k} \geq q_{k} \, ,\, k \neq k^{*} } \left( p_{k} - q_{k} \right)}^{\geq 0} \\
        &\geq p_{k^{*}} \geq p_{\min} \, .
    \end{align}
    Finally, for any such $\{q_{i}\}$,
    \begin{align}
        \mathsf P_{\{q_{i}\}} &= \frac{V(p,q) p_{\min}}{V(p,q) + p_{\min}} \\
        &\geq \frac{p_{\min ^{2}}}{2p_{\min}} = \frac{1}{2}p_{\min} \, ,
    \end{align}
    where we used that $\mathsf P_{\{q_{i}\}}$ is strictly increasing in $V(p,q)$. Thus,
    \begin{equation}
        1 - \max_{\{q_{i}\}} \mathsf P_{\{q_{i}\}} \leq 1 - \frac{1}{2}p_{\min} \, ,
    \end{equation}
    and the lemma is proven.
\end{proof}

\subsection{Further bounds: uniform priors and non-extremal states}
\label{app:UniformAndConvexlyDep}
{In this appendix we state and prove two refinements of Theorem~\ref{BoundsLinDepSD} used in the main text: a bound for uniform priors  and a strengthening for sets containing non-extremal states.}

{A case that is often of interest in SD tasks is the one in which the prior is uniform, for which the bound takes a form that depends only on the number of states and the dimension of their span.}
\begin{restatable}[SD tasks with Uniform prior]{corollary}{CorUniform}
\label{corollary: UniformPrior}
Consider a set $\{s_k\}_{k=1}^N\subset\bar\Omega$ of states in a GPT system, and let $D:=\dim(\mathsf{Span}[\{s_k\}])$.Suppose $\{s_k\}$ is linearly dependent (so $N\geq D+1$). Consider an SD task with a uniform prior over $N$ states. Then
 \begin{align}
 \label{eq:BoundUniformPrior}
P_{\mathrm{guess}}\ \le\ 1-\frac{1}{N}\!\left(\frac{N-D}{N-D+1}\right).
\end{align}
\end{restatable}

\begin{proof}
    For $D = \mathrm{dim}\left( \mathrm{Span}\left( \{q_{k}\} \right) \right)$, we have a linear decomposition of $\overline{s}$ with $q_{k}>0$ for \emph{at most} $D$ values of $k$. Therefore, there exists \emph{at least} $N-D$ values of $k$ with $q_{k}=0$. Using that $p_{k} = \frac{1}{N}$ for all $k$, we may then write
    \begin{align}
        V(p,q) &= \sum_{k \,  | \, p_{k} \geq q_{k} } \left( p_{k} - q_{k} \right) \\
        &= \sum_{k \,  | \, p_{k} \geq q_{k} \, , q_{k}=0 } \left( p_{k} - \overbrace{q_{k}}^{=0} \right) + \overbrace{\sum_{k \,  | \, p_{k} \geq q_{k} \, , q_{k} \neq 0} \left( p_{k} - q_{k} \right)}^{\geq 0} \\
        &\geq \sum_{k \,  | \, p_{k} \geq q_{k} \, , q_{k}=0 }  p_{k} =  \sum_{k \,  | \, p_{k} \geq q_{k} \, , q_{k}=0 }  \frac{1}{N} \\
        &\geq \frac{N-D}{N} \, .
    \end{align}
    Moreover, $p_{\min} = \frac{1}{N}$ and the value of $\mathsf P_{ \{q_{k}\} }$ is
    \begin{align}
         \mathsf P_{\{q_{i}\}} &= \frac{V(p,q) p_{\min}}{V(p,q) + p_{\min}} \\
         &\geq \frac{ \frac{N-D}{N} \cdot\frac{1}{N} }{\frac{N-D}{N} + \frac{1}{N}} = \frac{1}{N} \left(\frac{N - D}{N - D + 1} \right) \, .
    \end{align}
    The corollary then follows from direct application of Theorem~\ref{BoundsLinDepSD}.
\end{proof}

Another special case of SD tasks occurs when some of the states to be discriminated are mixtures of other states in the set $\{s_k\}$ (ie, they are not extremal in $\mathrm{ConvHull}[\{s_k\}]$).
The next result then gives a new bound for those cases, which depends only on the prior $\{p_k\}$ and the ability to discriminate among the extremal points of $\{s_k\}$ (with some induced prior).
\begin{lemma}[Lower bound on $P_{\mathrm{error}}$ for non-extremal subsets]
\label{lemma:ConvexlyDependent}
Fix a finite set $\{s_k\}\subset\bar{\Omega}$ of states in some GPT system, and let $\bar S=\operatorname{ext}[\{s_k\}]$ be its set of extremal points in $\operatorname{ConvHull}[\{s_k\}]$. For any SD task that uses all states in the set, if $\{s_k\}$ contains at least one non-extremal state (i.e.\ $\bar S\subsetneq\{s_k\}$), then
\begin{align}
P^{\{s_k\}}_{\mathrm{error}}
\ \ge\
\frac{p_{\min}}{\,p_{\mathrm{interior}}+p_{\min}\,}\Bigl(p_{\mathrm{interior}}+P^{\bar{S}}_{\mathrm{error}}\Bigr),
\end{align}
where {$p_{\mathrm{interior}} := \sum_{s_i \in \{s_k\}\setminus \bar S} p_i$ is the total prior assigned to the non-extremal states}, $P^{\bar S}_{\mathrm{error}}$ is the optimal error probability of the SD task restricted to the extremal states $\bar S$ (with the induced prior), and $p_{\rm min}:=\min_k{p_k}$.
\end{lemma}

Combining Lemma~\ref{lemma:ConvexlyDependent} with Corollary~\ref{corollary:MaxAndHalfBounds} is an interesting strategy to calculate bounds whenever some states in the original SD task are mixtures of others: one first reduces the original SD problem containing non-extremal states to another with only extremal points, and then uses Corollary~\ref{corollary:MaxAndHalfBounds} to get a bound on $P^{\bar S}_{\rm error}$.

\begin{proof}%
We write our set of states $\{s_{k}\}$ as a disjoint union of sets $S$ and $S'$ such that $S$ is the set of extremal states in $\text{ConvHull}(\{s_{k}\})$ and $S'$ is the set of non-extremal states.

With this construction, every point in $S'$ can be written as a convex combination of states in $S$. That is, for every state $s_{j} \in S'$, one may write
\begin{equation}
    s_{j} = \sum_{i|s_{i}\in S} \lambda_{i}^{j} s_{i}
\end{equation}
for some set $\{\lambda_{i}^{j} \}_{i}$ of non-negative real numbers with $\sum_{i}\lambda_{i}^{j}=1$.

For every $i$ such that $s_{i} \in S$, we thus define
\begin{equation}
    q_{i} = p_{i} + \sum_{j| s_{j} \in S'}p_{j}\lambda_{i}^{j} \,. \label{qidefinition}
\end{equation}
One can readily check that the $q_{i}$  define a probability distribution --- i.e., $q_{i} \geq 0$ and $\sum_{i}q_{i} = 1$ --- and moreover
\begin{equation}
    \overline{s} = \sum_{i}q_{i}s_{i} \, .
\end{equation}
We also note that from \eqref{qidefinition} we have that, since the $\lambda_{i}^{j}$ are all non-negative, $q_{i} \geq p_{i}$ for all $i$ such that $s_{i} \in S$ and therefore $Q_{+}=\{i|s_{i} \in S\}$.

From equation \eqref{boundwithsum}, it thus follow that
\begin{equation}
    P_{\text{error}}^{\{s_{k}\}} \geq \sum_{k \in Q_{+}}(q_{k}-p_{k})B(e_{k},s_{k}) + \sum_{k \in Q_{+}} \sum_{i\neq k}q_{i} B(e_{k},s_{i}) \, .
\end{equation}

Since $q_{i} = 0$ for all $i \notin Q_{+}$, the rightmost term is
\begin{equation}
    \sum_{k \in Q_{+}} \sum_{i\neq k}q_{i} B(e_{k},s_{i}) = \sum_{k \in Q_{+}} \,\sum_{i \in Q{+}|i\neq k}q_{i} B(e_{k},s_{i}) \, ,
\end{equation}
which is the error probability of the state discrimination task $(q_{i},s_{i})_{i \in Q_{+}}$ with a (potentially non-optimal) measurement. As such, this quantity is $\geq P_{\text{error}}^{S}$.

For the leftmost term, we again use \eqref{ineq:SDBoundGeneral} to get
\begin{equation}
    B(e_{k},s_{k}) \geq \ubar{\lambda} \geq \left( 1 - \frac{P_{\text{error}}^{\{s_{k}\}} }{p_{\min}} \right)\, \, .
\end{equation}
As such, we arrive at
\begin{equation}
    P_{\text{error}}^{ \{ s_{k} \} } \geq \left( 1 - \frac{P_{\text{error}}^{\{s_{k}\}} }{p_{\min}} \right) V(\{q_{k}\},\{p_{k}\}) + P_{\text{error}}^{S} \, . \label{felipe1}
\end{equation}
Now, we note that the total variation distance takes the form
\begin{align}
    V(\{q_{k}\},\{p_{k}\}) &= \sum_{i \in Q_{+}} \lvert q_{i}-p_{i} \lvert \\[3pt]
    &= \sum_{i \in Q_{+}} \sum_{j \notin Q_{+}} p_{j} \lambda_{i}^{j} \\[3pt]
    &=\sum_{j \notin Q_{+}}p_{j} \,  \equiv p_{\text{interior}} \, .
\end{align}
Finally, plugging this into equation \eqref{felipe1} and rearranging gives us the desired result
\begin{equation}
     P_{\text{error}}^{ \{ s_{k} \} } \geq  \frac{p_{\min}}{p_{\min} + p_{\text{interior}}} \left( p_{\text{interior}} + P_{\text{error}}^{S} \right) \, .
\end{equation}
\end{proof}

\begin{example}
\label{example:cube}
Consider an SD task whose set of states $\{s_k\}$ consists of the eight vertices of a cube, with uniform prior. The average state is the center of the cube, and any antipodal pair (e.g. $s_1$ and $s_5$) provides a linearly independent decomposition of the center with $q_1=q_5=\tfrac12$ and $q_i=0$ otherwise (see Fig. ~\ref{fig:cube}). Hence $V(\{p_k\},\{q_k\})=\tfrac{6}{8}$, and Theorem~\ref{BoundsLinDepSD} gives
\begin{equation}
P_{\mathrm{guess}} \;\le\; 1-\frac{V(\{p_k\},\{q_k\})}{V(\{p_k\},\{q_k\})+\tfrac{1}{8}}\cdot \frac{1}{8}
\;=\; \frac{25}{28}\;\approx\; 0.893.
\end{equation}
This bound is purely geometrical and applies to any realization of this cube hull. In particular, one may view $\{s_k\}$ as eight pure qubit states forming a cube in the Bloch ball, or as the eight extremal states of a three-dimensional gbit~\cite{Janotta_2013GPTsWihtoutNH}; in both cases the same numerical bound holds (though these are likely not tight). In fact, any SD task $\{p_k,t_k\}$ in which $t_k=m(s_k)$ where $m$ is an injective linear map will necessarily respect the same bound.

Corollary~\ref{corollary: UniformPrior} and Corollary~\ref{corollary:MaxAndHalfBounds} allow for simpler calculations, though slightly looser bounds. Using that $\mathsf{Span}[\{s_k\}]=4$ and $N=8$ gives $P_{\rm guess}\leq 0.9$ for the former, and using $p_{\rm min}=\frac{1}{8}$ provides $P_{\rm guess}\leq \frac{7}{8}\approx 0.938$.

Now consider another SD task: extend $\{s_k\}$ by adding the six face-centers of the cube (total $N=14$ points) and keep the prior uniform. Here,  the extremal set $\bar S$ are the eight cube vertices (uniform by symmetry), thus $p_{\min}=\tfrac{1}{14}$ and the total prior probability for non-extremal points is $p_{\mathrm{interior}}=\tfrac{6}{14}$. Combining Lemma~\ref{lemma:ConvexlyDependent} with the above bound for $\bar S$ yields
\begin{align}
&P_{\mathrm{error}}^{\{s_k\}} \;\ge\; \frac{1}{7}\cdot\frac{3}{7} + \frac{1}{7}\cdot \frac{3}{28}\\
\therefore\,\,\,\,
&P_{\mathrm{guess}} \;\le\; 1-\Big(\frac{3}{7^2}+\frac{3}{7\cdot28}\Big)\;\approx\; 0.923.
\end{align}
By invariance under linear maps (Appendix~\ref{app:affine-invariance}, Corollary~\ref{corollary:invarianceBoundsINjectiveLinearM}; see also Example~\ref{ex:diagonal-stretch}), any set linearly equivalent to either of these geometries obeys the same numerical bounds. See Fig.~\ref{fig:cube}.
\end{example}

\subsection{Proof of Corollary~\ref{corollary:invarianceBoundsINjectiveLinearM}: Invariance of the geometric SD bounds}
\label{app:affine-invariance}

\begin{lemma}[Injective linear maps preserve linearly independent decompositions]
Consider an ensemble of GPT states $\{(p_k,s_k)\}_{k\in K}$ and let $L:\mathrm{Span}[\{s_k\}_{k\in K}]\to W$ be an injective linear map. Then, $\sum_k q_ks_k$ is a linearly independent decomposition for $\bar{s}$ if and only if $\sum q_k L(s_k)$ is a linearly independent decomposition for $L(\bar{s})$. Explicitly,
\begin{align}
    \bar s=\sum_{i\in S} q_k s_k \Longleftrightarrow\quad
L(\bar s)=\sum_{i\in S} q_k L(s_k),
\end{align}
with $\{s_k\}_{k\in {\rm Supp}[\{q_k\}]}$ linearly independent if and only if  $\{L(s_k)\}_{k\in {\rm Supp}[\{q_k\}]}$ is linearly independent.
\end{lemma}
\begin{proof}
     Due to linearity of $L$, we have two consequences:
     \begin{enumerate}
         \item[i)] $\bar{s} = \sum p_ks_k \implies L(\bar{s})= \sum p_k L(s_k)$;
         \item [ii)] For any $\{s_i\}\subset\mathrm{Span}[\{s_k\}]$,  $\{L(s_i)\}$ is LI implies that $\{s_i\}$ is LI. ($L$ cannot ``create'' linear independence).
     \end{enumerate}
     Now, if $L$ is furthermore injective, $L(v)=0\implies v=0$. Therefore, we get the converses:
     \begin{itemize}
         \item[iii)] $L(\bar{s})= \sum p_k L(s_k) \implies\bar{s} = \sum p_ks_k$;
         \item[iv)] For any $\{s_i\}\subset\mathrm{Span}[\{s_k\}]$,  $\{s_i\}$ is LI implies $\{L(s_i)\}$ is LI ($L$ preserves independence).
     \end{itemize}
     Properties $i)$ and $iv)$ imply that if $\{q_k\}$ is a linearly independent decomposition  of $\bar s$ from $\{s_k\}$, then $\sum q_k L(s_k)$ is a linearly independent decomposition of $L(\bar{s})$. Properties $ii)$ and $iii)$ imply that if $\sum q_kL(s_k)$ is a linearly independent decomposition of $L(\bar s)$, then $\sum q_ks_k$ is a linearly independent decomposition of $\bar s$.
\end{proof}

Consequently, any numerical bound that depends only on the priors $(p_k)$ and the weighs $q_k$ of linearly independent decompositions of the average state  (e.g., $p_{\min}$ and the total variation distance $V(\{p_k\},\{q_k\})$) is unchanged when passing from a task $\{(p_k,s_k\})$ to the task $\{(p_k,L(s_k))\}$.
This leads to the corollary that we wanted to prove:
\BoundInvarianceInjective*

{\begin{rmk}
\label{rmk:InvarianceAppendixBounds}
The argument above covers every bound built from the priors and from linearly independent decompositions of the average state, i.e., Theorem~\ref{BoundsLinDepSD} and Corollary~\ref{corollary:MaxAndHalfBounds}. The two results of Appendix~\ref{app:UniformAndConvexlyDep} are invariant as well. For Corollary~\ref{corollary: UniformPrior}, it suffices to note that an injective linear map preserves $D=\dim(\mathsf{Span}[\{s_k\}])$, while $N$ and the (uniform) prior are untouched. For Lemma~\ref{lemma:ConvexlyDependent}, an injective linear map is a linear isomorphism onto its image and hence preserves the convex structure of $\mathrm{ConvHull}[\{s_k\}]$ (in particular, which states are extremal) while the priors, and thus $p_{\rm min}$, $p_{\rm interior}$ and the induced prior on the extremal set, are unchanged. Hence the numerical bounds of both results coincide for $\{(p_k,s_k)\}$ and $\{(p_k,M(s_k))\}$.
\end{rmk}}

As a quick remark, notice that we have been using the term `linearly equivalent' instead of affinely equivalent maps, while these maps also preserve the convex structure required for the bounds. The reason is that, although we are here only concerned about normalized states, ultimately any map that takes GPT states in $\bar{\Omega_G}$ to GPT states in $\bar{\Omega}_H$ (or state maps, in the terms of Ref.~\cite{Schmid_ShadowsSubsystems2025}) should map the $0$ vector in $\mathsf{Span}[\bar{\Omega}_G]$ to the $0$ vector in $\mathsf{span}[\bar{\Omega}_H]$, thus being linear, to also preserve the subnormalized states structure.

\begin{example} \label{ExampleDiagonalStretch}

[Diagonal stretch of a square leaves the bound unchanged]\label{ex:diagonal-stretch}
Let $v_1=(1,1,k)$, $v_2=(-1,1,k)$, $v_3=(1,-1,k)$, $v_4=(-1,-1,k)$ in an affine subspace of $\mathbb{R}^3$ with fixed $k\neq 0$.
Consider any prior $\{p_i\}_{i=1}^4$ and average $\bar s=\sum_i p_i v_i$.
Let $A$ be the invertible affine (here linear) map that scales along the $(1,1)$ diagonal by a factor of $2$ and leaves the $(1,-1)$ diagonal fixed:
$A(v_1)=(2,2,k)$, $A(v_2)=(-1,1,k)$, $A(v_3)=(1,-1,k)$, $A(v_4)=(-2,-2,k)$.
Then $A(\bar s)=\sum_{j\in S} q_j\,A(v_j)$ with the same coefficients $\{q_j\}_{j\in S}$.
Thus $V(\{p_k\},\{q_k\})$ is unchanged and the Theorem~\ref{BoundsLinDepSD}/Corollary~\ref{corollary:MaxAndHalfBounds} bound is identical before and after the transformation.
{The same numerical bound therefore applies not only to this stretched square, but to any realization of the square that is linearly equivalent to it -- for instance, as the four extremal states of a gbit (Fig.~\ref{fig:GPTSystemsExample}) or as four qubit states forming a square in an equatorial plane of the Bloch ball.}
\end{example}

\subsection{Proof of Lemma~\ref{lemma:CandLinearDep}}
\label{app:ProofLemmaCandLD}
In order to prove {Lemma~\ref{lemma:CandLinearDep}}, we first introduce three important ingredients. First, we recall the notion of a strictly classical GPT system (following Ref.~\cite{Schmid_ShadowsSubsystems2025}, which are simplicial GPT systems.
\begin{definition}[Simplicial GPT system = strictly classical GPT systems]
A simplicial GPT system is a system $G^d_{cl} = (\Delta_n, \Delta^{*}_n, p_{\Lambda})$
in which the normalized state space \( \Delta \) is a simplex in a finite-dimensional real vector space of dimension $n$, its effect space \( \Delta^{*} \) is the logical dual of \( \Delta \), consisting of all vectors $e\in \mathbb{R}^n$ such that  the inner product with $s$ is in the interval $[0,1]$ for all  $s \in \Delta$, and the probability rule \( p_{\Lambda} \) coincides with the Euclidean inner product \( p_{\Lambda}(e,s) = \langle e, s \rangle \) between \( \Delta^{*} \) and \( \Delta \). In other words, a simplicial GPT system has a simplex as state space, and the effect space can be thought of as being the hypercube dual to that simplex.
\end{definition}

Secondly, we introduce GPT fragment embeddings.
\begin{definition}[GPT fragment embedding]
    Consider a GPT system $(\Omega,\mathcal{E},B(.,.))$ and a fragment of thereof, ie, $f:=(\Omega',\mathcal{E}',B(.,.))$ with $\Omega'\subset\Omega$, $\mathcal{E}'\subset{\mathcal{E}}$. The fragment $f$ is said to embed into another GPT system $(\Sigma,\mathcal{F},C(.,.))$ if there exist two injective maps $\iota:\mathsf{Span}[\Omega']\to\mathsf{Span}[\Sigma]$ and $\kappa:\mathsf{Span}[\mathcal{E}']\to\mathsf{Span}[\mathcal{F}]$ which together reproduce the probabilities:
    \begin{align}
        C(\kappa(e),\iota(s))=B(e,s)\,\,\forall s\in\Omega'\,,e\in\mathcal{E}'.
    \end{align}
\end{definition}
Finally, Ref.~\cite{Schmid_2021_NCinGPTsystems} showed that for single GPT systems, \emph{noncontextuality is equivalent to the existence of an embedding into a simplicial system}. This remains true for a subset of states and effects (called fragments), as shown in Ref.~\cite{Selby_2023}.
\begin{definition}[Embedding into simplicial GPT systems]
    Consider a GPT system $(\Omega,\mathcal{E},B(.,.))$. A GPT fragment $f:=(\Omega',\mathcal{E}',B(.,.))$, with $\Omega'\subset\Omega$ and $\mathcal{E}'\subset \mathcal{E}$, is simplex-embeddable if and only if there exists an embedding $\iota,\kappa$ of $f$ into a simplicial GPT system $(\Delta_n,\Delta^*_n,p_{\Lambda})$ for some $n\in\mathbb{N}$.
\end{definition}
\LinearIndepSufficientForNC*
\begin{proof}

     Consider a GPT system $(\Omega,\mathcal{E},B(.,.))$ and a linearly independent subset of states $\{s_k\}\subset\Omega$.
    Complete the set $\{s_k\}_{k=1}^N$ to a set $\{s'_k\}_{k=1}^D$ so that  $\{s'_k\}\subset\Omega$ remains linearly independent and, moreover, span the full vector space where $\Omega$ lives, ie, $\mathsf{Span}[\{s'_k\}]=\mathsf{Span}[\Omega]$. Note that this an always be done, and if the original set already spans the space, this completion lives $\{s_k\}$ unchanged. Note that the fragment $(\{s_k\},\mathcal{E},B(.,.))$ always embeds into $(\Omega',\mathcal{E},B(.,.))$. We will show that the fragment that has all original effects $\mathcal{E}$ and the subset of states $\Omega':=\mathsf{ConvHull}[\{s'_k\}]$ is simplex embeddable.

  Define $\mathcal{E}^{\rm max}$ as the set of all elements of $\mathsf{Span}[\mathcal{E}]$ such that $B(e,s)\in[0,1]$ for all $s\in\Omega'$. Since $\Omega'\subset \Omega$, then $\mathcal{E}^{\rm max}$ contains the original set of effects,ie  $\mathcal{E}\subset\mathcal{E}^{\rm max}$. Thus, $(\Omega',\mathcal{E},B(.,.))$ trivially embeds into $(\Omega',\mathcal{E}^{\rm max},B(.,.))$ via an inclusion map. Note that since $\Omega'$ is a simplex, we can define a map $\iota$ which maps each element of $\{s'_k\}$ into a different vertex of a $\Delta_{D-1}$ simplex. Since $\mathcal{E}^{\rm max}$ is the set of all effects that map points in $\Omega'$ to $[0,1]$ via a nondegenerate probability rule $B(.,.)$, the set  $\{B(e,.)|e\in\mathcal{E}^{\rm max}\}$ is isomorphic to the set $p_{\rm Lambda}(f,.)$ in which $f\in\Delta^*$. Therefore, there exists an injective map $\kappa$ that maps $\mathcal{E}^{\rm max}$ to $\Delta_{D-1}^*$. Thus, $(\Omega',\mathcal{E}^{\rm max},B(.,.))$ embeds into $(\Delta_n,\Delta^*_n,p_{\rm \Lambda})$ for $n=D-1$.

  We have thus showed that $(\{s_k\},\mathcal{E},B(.,.))$ embeds into $(\Omega',\mathcal{E},B(.,.))$, which in turn embedds into $(\Omega',\mathcal{E}^{\rm max},B(.,.))$ which finally embeds into $(\Delta_{D-1},\Delta^*_{D-1},p_{\rm \Lambda})$. Since embddings are transitive~\cite{Schmid_ShadowsSubsystems2025}, this proves $(\{s_k\},\mathcal{E},B(.,.))$ is simplex-embeddable.
\end{proof}

\subsection{Proof of Corollary~\ref{Corollary:gOM}}
\label{app:ProofCorollaryGOM}

\CorGOM*
\begin{proof}
Assume that the set $\{s_x\}_{x\in\{0,\dots,d-1\}^n}$ provides a nonclassical advantage in the $d$-ary oblivious multiplexing task. Since advantage in generalized oblivious communication tasks requires generalized contextuality, Theorem~\ref{theorem:linkCandSDtasks} implies that the set $\{s_x\}$ must obey the bounds of Theorem~\ref{BoundsLinDepSD} in any state-discrimination task using all the states in the set. Applying Corollary~\ref{corollary: UniformPrior} to the uniform-prior SD task on $\{s_x\}$, with $N=d^n$ and $D:=\dim(\mathrm{Span}[\{s_x\}_x])$, yields
\begin{align}
P_{\rm guess}
<
1-\frac{1}{d^n}\left(\frac{d^n-D}{d^n-D+1}\right)
<
1-\frac{1}{2d^n},
\end{align}
which proves the first claim.\footnote{{As discussed below Example~\ref{example:gOM} in the main text, the obliviousness constraints of Eq.~\eqref{eq:obliviousness} alone already force linear dependence of the full codebook, so this first claim in fact holds for every suitable encoding, with or without nonclassical advantage.}}

For the second claim, recall that for each $y\in\{1,\dots,n\}$, Bob's subtask is the SD problem with ensemble
\begin{align}
\left\{\frac{1}{d},\, s_{x_y}\right\}_{x_y=0}^{d-1},
\qquad
s_{x_y=v}:=\frac{1}{d^{\,n-1}}\sum_{x:\,x_y=v}s_x,
\end{align}
and that the total success probability satisfies
\begin{align}
P_{\rm succ}^{\rm OC}=\frac{1}{n}\sum_{y=1}^n P_{{\rm guess},y}.
\end{align}
Suppose now that
\begin{align}
P_{\rm succ}^{\rm OC}>1-\frac{1}{2d}.
\end{align}
Since $P_{\rm succ}^{\rm OC}$ is the average of the $n$ numbers $P_{{\rm guess},y}$, there must exist some $y^*\in\{1,\dots,n\}$ such that
\begin{align}
P_{{\rm guess},y^*}>1-\frac{1}{2d}.
\end{align}
Now consider the corresponding sub-ensemble
\begin{align}
\left\{\frac{1}{d},\, s_{x_{y^*}}\right\}_{x_{y^*}=0}^{d-1}.
\end{align}
If this sub-ensemble could provide an advantage in a contextuality-powered task, then by Theorem~\ref{theorem:linkCandSDtasks} it would have to obey the bounds of Theorem~\ref{BoundsLinDepSD} in any SD task using all its states. In particular, by Corollary~\ref{corollary:MaxAndHalfBounds} applied to the uniform prior $1/d$, one would have
\begin{align}
P_{{\rm guess},y^*}\le 1-\frac{1}{2d},
\end{align}
which contradicts the inequality above. Therefore, the sub-ensemble
\begin{align}
\left\{\frac{1}{d},\, s_{x_{y^*}}\right\}_{x_{y^*}=0}^{d-1}
\end{align}
cannot provide advantage in any contextuality-powered task. This proves the result.
\end{proof}

\end{document}